\documentclass[format=acmsmall, review=false, screen=true]{acmart-arXiv}

\usepackage{pgfplots}
\usepackage{subfigure}
\usepackage{graphicx,color}%
\usepackage{graphics}
\usepackage{listings}
\usepackage{algorithm}
\usepackage{algpseudocode}
\usepackage{framed}
\usepackage{graphicx}
\usepackage[mathscr]{euscript}
\usepackage{pdfpages}
\usepackage[font=small,skip=0pt]{caption}

\allowbreak

\def\b1{\mathbf{1}}

\newcommand{\cG}{\mathcal{G}}

\newcommand{\E}{\mathbb{E}}
\newcommand{\cE}{\mathcal{E}}

\newcommand{\sM}{\mathscr{M}}
\newcommand{\N}{\mathbb{N}}
\newcommand{\Prob}{\mathbb{P}}
\newcommand{\Nbr}{\mathscr{N}}

\usepackage{bbm}
\newcommand{\ind}[1]{\mathbbm{1}_{\{#1\}}}
\usepackage{prettyref}
\newrefformat{sec}{Section~\ref{#1}}
\newrefformat{subsec}{Section~\ref{#1}}
\newrefformat{subsubsec}{Section~\ref{#1}}
\newrefformat{thm}{Theorem~\ref{#1}}
\newrefformat{THM}{THEOREM~\ref{#1}}
\newrefformat{lemma}{Lemma~\ref{#1}}
\newrefformat{prop}{Proposition~\ref{#1}}
\newrefformat{LEMMA}{LEMMA~\ref{#1}}
\newrefformat{corollary}{Corollary~\ref{#1}}
\newrefformat{fig}{Figure~\ref{#1}}
\newrefformat{tab}{Table~\ref{#1}}
\newrefformat{eq}{\eqref{#1}}
\newrefformat{app}{Appendix~\ref{#1}}

\usepackage{tikz}
\usetikzlibrary{arrows.meta,shapes,graphs,graphs.standard,quotes}
\tikzset{>=latex}
\tikzstyle{vertex}=[circle,fill=black!25,minimum size=20pt,inner sep=0pt]
\tikzstyle{arc} = [draw,thick]
\tikzstyle{weight} = [font=\small]

\pgfplotsset{
  every axis plot/.append style={line width=1.5pt},
}

\begin{document}
\title{On a Class of Stochastic Multilayer Networks}


\author{Bo Jiang}
\affiliation{%
 \department{College of Information and Computer Sciences}
  \institution{University of Massachusetts}
  \city{Amherst} 
  \state{MA} 
  \postcode{01003}
  \country{USA}
}
\email{bjiang@cs.umass.edu}

\author{Philippe Nain}
\affiliation{%
  \institution{Inria} 
  \streetaddress{Ecole Normale Sup\'erieure de Lyon, LIP, 46 all\'ee d'Italie}
  \city{69364 Lyon} 
  \country{France} 
}
\email{philippe.nain@inria.fr}

\author{Don Towsley}
\affiliation{%
\department{College of Information and Computer Sciences}
  \institution{University of Massachusetts}
    \city{Amherst} 
  \state{MA} 
  \postcode{01003}
  \country{USA}
  }
\email{towsley@cs.umass.edu}

\author{Saikat Guha}
\affiliation{
\department{College of Optical Sciences}
  \institution{University of Arizona}
 \city{Tucson}
 \state{AZ}
 \postcode{85721}
 \country{USA}
}
\email{saikat@optics.arizona.edu}

\renewcommand{\shortauthors}{Bo Jiang et al.}

\begin{abstract}

In this paper, we introduce a new class of stochastic multilayer networks. 
A stochastic multilayer network is the aggregation of $M$ networks (one per layer) where each is a subgraph of a foundational network $G$. Each layer network is the result of probabilistically removing links and nodes from $G$. The resulting network includes 
any link that appears in at least $K$ layers. This model is an instance of a non-standard site-bond percolation model.
 Two sets of results are obtained: first, we derive the probability distribution that the $M$-layer network is in a given configuration for some particular graph structures (explicit results are provided for a line and an algorithm is provided for a tree), where a configuration is the collective state of all links (each either active or inactive). Next, we show that for appropriate scalings of the node and link selection processes
in a layer,  links are asymptotically independent as the number of layers goes to infinity, and follow Poisson distributions. Numerical results are provided to highlight the impact of having several layers on some metrics of interest (including expected size of the cluster a node belongs to in the case of the line).
This model finds applications in wireless communication networks with multichannel radios, multiple social networks with overlapping memberships, transportation networks, and, more generally, in any scenario where a common set of nodes can be linked via co-existing means of connectivity.

\end{abstract}

%
%




\maketitle


\section{Introduction}
\label{sec:intro}

There is an increasing need to understand how different networks interact with each other. One means of such interaction arises when users (nodes) belong to two or more networks (layers). In recent years, there has been a surge of interest in such {\em multilayer networks}~\cite{Kiv14, Boc14} due to their relevance in problems stemming in varied fields such as multifrequency wireless communication networks~\cite{Red11}, multiple online social networks serving a common population~\cite{Sze10, Men14, Bas15} just to name a few. Various models of multilayer networks (also termed multiplex networks and composite networks in the literature) relevant to different application scenarios have been proposed, in particular stochastic multilayer networks whose constructions can be described by one or more control parameters (such as probability of the presence of a node, edge or more complex attributes). For such networks, a wide variety of percolation formulations have been proposed and studied, e.g., competition between layers~\cite{Zha13},
weak percolation~\cite{Bax14}, $k$-core percolation~\cite{Azi14},
directed percolation~\cite{Azi14a}, spanning connectivity of a multilayer site-percolated network~\cite{Guh16}, and bond percolation~\cite{Hac16}. However, even simple multilayer network models have proven extremely difficult to analyze exactly~\cite{Guh16}. Consequently, most of this aforesaid recent literature on properties of multilayer networks consists of numerical and heuristic analyses.

Our goal in this paper is to consider a simple model for a stochastic multilayer network and to attempt exact characterization of the joint probability distribution of the collective (on-off) configuration of the links of the multilayer network. We provide exact results and efficient algorithms for some special graphs, and prove some complexity-theoretic hardness results in the general case. Our model is as follows. A multilayer network consists of $M$ co-existing networks $G^{(1)}$, $G^{(2)},\ldots, G^{(M)}$ connecting a common set of users. Each user is active in only a subset of these networks. We also say a user active in a particular network belongs to that network. A user active in both $G^{(1)}$ and $G^{(2)}$, for example, can help connect two other users that are active in $G^{(1)}$ alone, and in $G^{(2)}$ alone, respectively, by forming a bridge. \prettyref{fig:multilayer-network} illustrates an example with $M=3$ networks (layers), where a path connecting $v_1$ and $v_2$ must traverse all three layers, and one such path is shown to go through the bridge nodes $v_3$ and $v_4$, both of which belong to more than one layer.  A {\em stochastic multilayer network} is a graph $G=(V,E)$ along with a random process by which each network layer is obtained from $G$ by randomly removing links (called {\em link thinning}) and randomly deactivating nodes, and a process by which the $M$ thinned layers are merged into a single graph. Layer $m$, $1 \le m \le M$ is a subgraph of $G$ consisting of all remaining active nodes and all links between active nodes not removed through link thinning. There are different ways of creating a multilayer network out of the $M$ layers. One is simply to take the union of (the nodes and links of) all the layer graphs; this is illustrated by the three layer network in Figure~\ref{fig:multilayer-network}. We consider a slightly more general process whereby all active nodes are included in the final graph and all links that appear in at least $K$ layers. 

\begin{figure}[ht]
\begin{center}
\captionsetup{skip=3pt}
\includegraphics[width=0.5\columnwidth]{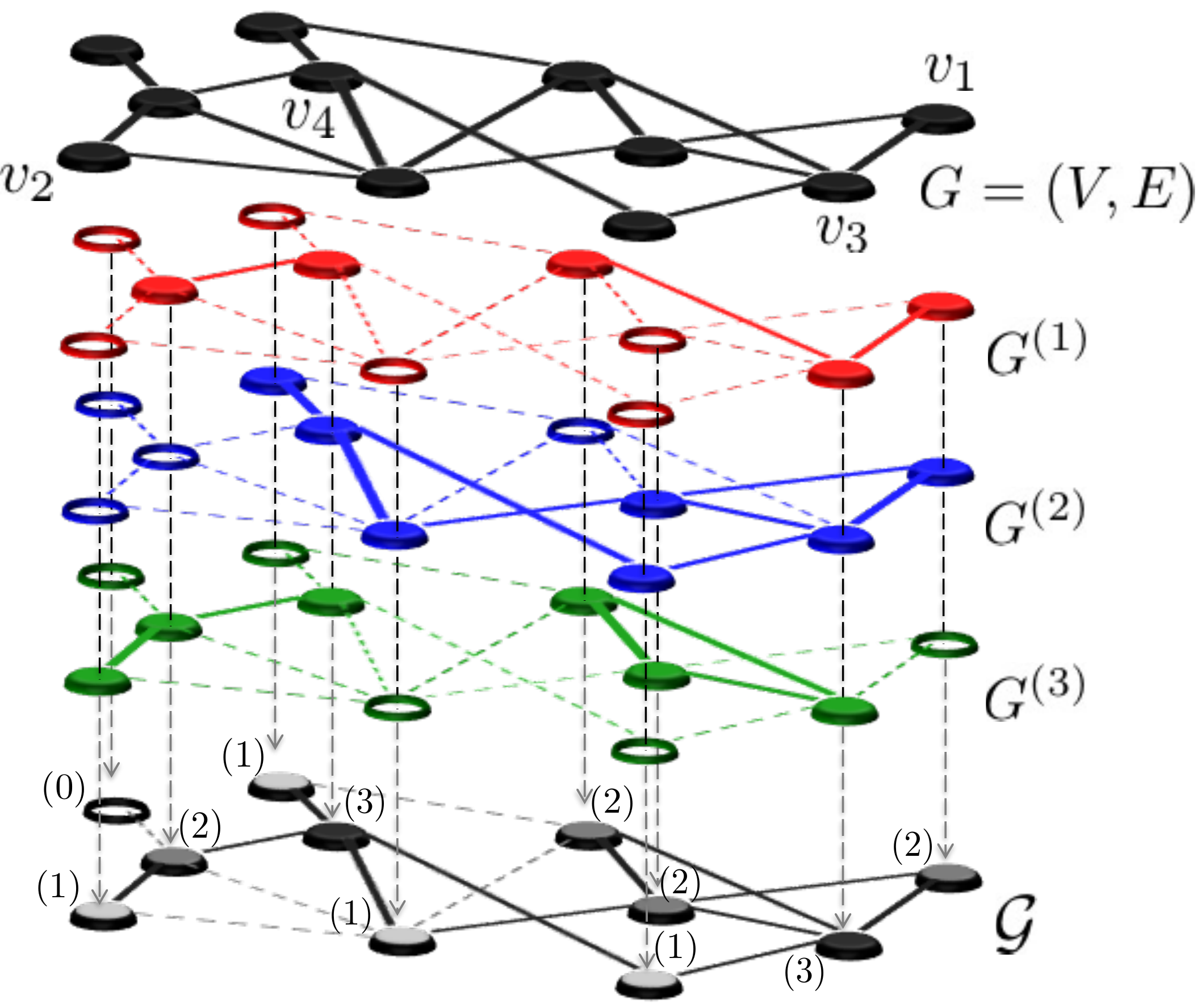}
\caption{Multilayer network with three layers.}
\label{fig:multilayer-network}
\end{center}
\end{figure}

Concrete examples of such multilayer networks are: (1) a network of cities connected via different airline companies where each city is served only by a subset of all the airlines~\cite{Bas15,Buc13}, (2) a network of users with accounts on multiple online social networks~\cite{Mur14}, and (3) a communication network of units equipped with radios that can listen and transmit simultaneously on a subset of multiple frequencies~\cite{Red11}. 

With a more liberal interpretation of ``co-existence'', such multilayer networks may also arise from taking snapshots of a single network at different time epochs. For example, consider a duty-cycled wireless sensor network where each sensor is active or dormant according to a random periodic schedule and each period is divided into $M$ slots. The $m$-th layer $G^{(m)}$ then consists of sensors that are active in the $m$-th slot. Duty-cycled models have been studied in the wireless sensor network literature (e.g.,~\cite{bagchi2015optimal} and references therein), where the underlying networks are usually random geometric graphs and the focus is on the connectivity of each layer, which is a much stronger notion of connectivity than  connectivity in our aggregate network.

Denote the configuration of a multilayer network by the collection of states of all links in the underlying graph $G$ after the layers have been merged.  Here the state of each link is either active (1) or inactive (0). We are interested in characterizing the configuration probability distribution of the multilayer network under the assumption that thinning and deactivation operations occur as  independent events. Such a characterization can be useful for computing quantities such as the distribution of the sizes of connected components and average path lengths. We show that in general, computing the network configuration distribution is hard - for example, in most cases computing the probability that there are no active links in the merged graph is \#P-hard. On the other hand we have partial positive results for some classes of graphs including trees. Moreover, we consider the behavior of this distribution in the limit as $M\rightarrow \infty$. Our contributions are:

\begin{itemize}
\item We present a new model of a stochastic multilayer network based on link thinning and node deactivation, and show that in general it is a difficult problem to compute probabilities of multilayer network configurations and it remains difficult even to approximate these probabilities.

\item We develop efficient algorithms for computing multilayer network configuration probabilities for line and  tree topologies.

\item We consider a setting where the number of layers $M$ goes to infinity and where link thinning probabilities and node deactivation probabilities are functions of $M$.  We provide conditions for link existence events to be asymptotically independent. 
\end{itemize}

The paper is organized as follows. Section \ref{sec:model} presents our stochastic multilayer model. The hardness of the problem of computing multilayer network configuration probabilities is addressed in Section \ref{sec:hardness}. Exact results and efficient computational algorithms are presented in Section \ref{sec:exact} and the asymptotic independence of the link states as $M\rightarrow \infty$ is found in Section \ref{sec:asymp}. A discussion of related work can be found in Section \ref{sec:related} and conclusions are drawn in Section \ref{sec:concl}.


\section{Model}\label{sec:model}

Let $G =(V,E)$ denote the underlying connectivity network, where $V$ is the set of nodes and $E\subset V\times V$ is the set of links that represent all \emph{possible} connections between pairs of nodes in $V$ (in the graph/percolation community a node is called a vertex/site and a link is called an edge/a bond; throughout we will use node and link
which are commonly used in communication networks). We assume network $G$ is connected.

Consider an $M$-layer network whose layers are sub-networks of $G$ obtained by randomly removing links (called \emph{link thinning}) and deactivating nodes. When a node is deactivated on a layer, all links incident on it are removed from the same layer, including those that have survived the independent \emph{link thinning} process. More precisely, the $M$-layer network is obtained from $G$ as follows. Let $\sM=\{1,2,\dots,M\}$ be the index set for layers.  Let $Y_{\sM,E} = \{Y_{m,\ell}:\ell\in E, m\in \sM\}$ and $Z_{\sM,V} = \{Z_{m,i}:i\in V, m\in \sM\}$ be two mutually independent sets of independent Bernoulli random variables.  For the $m$-th layer $G^{(m)}=(V, E^{(m)})$, node $i$ is active if and only if $Z_{m,i} = 1$, and the link set $E^{(m)}$ is given by $E^{(m)} = \{\ell\in E: W_{m,\ell}  = 1\}$, where $W_{m,\ell}=Y_{m,(i,j)} Z_{m,i} Z_{m,j}$ for link $\ell=(i,j)$. Note that link $\ell=(i,j)$ is in $E^{(m)}$ if and only if it is not thinned ($Y_{m,\ell}=1$) and both endpoints $i,j$ are active on the $m$-th layer ($Z_{m,i} = Z_{m,j} = 1$). We assume that the link thinning probabilities and the node activation probabilities are the same across different layers but may depend on individual links and nodes, i.e.  $p_\ell = \E[Y_{m,\ell}]$ and  $q_i = \E[Z_{m,i}]$ for all $m\in \sM$, $\ell\in E$ and $i\in V$. This assumption will be relaxed in \prettyref{subsec:asymptotic-non-identical}.  We also assume that all $p_\ell$'s and $q_i$'s are strictly positive.

Let $W_\ell \in \N$ denote the number of layers in which link $\ell$ is present. Formally, for $\ell = (i,j)$,
\begin{equation}\label{eq:W}
W_\ell = \sum_{m\in \sM} W_{m,\ell}.
\end{equation}
Note that we have suppressed the explicit dependence on $M$ of all random variables for notational simplicity. No confusion should arise. Given some threshold $K\in \N$, let $X_\ell = \ind{W_\ell \geq K}$, where $\mathbbm{1}_{A}$ is the indicator of event $A$.
We say that link $\ell$ is \emph{active} (\emph{inactive}) in the multilayer network if $X_\ell =1$ ($X_\ell = 0$). We obtain a merged network $\cG=(V,\cE)$, where $\cE = \{\ell\in E: X_\ell=1\}$ is the set of active links; we say the multilayer network has link configuration $\cE$, or equivalently, configuration $X_E\triangleq\{X_\ell:\ell\in E\}$. We will use the terms configuration and state interchangeably. The parameter $K$ determines the robustness of links in $\cG$; a larger value of $K$ results in more robust links but a possibly less well connected network $\cG$. Note that when $K=1$, $\cG$ is simply the union of the layers, i.e. $\cG = \bigcup_{m=1}^M G^{(m)}= (V,\bigcup_{m=1}^M E^{(m)})$. More generally,  we call any vector $x$ with component $x_\ell\in \{0,1\}$ for $\ell\in E$ a link configuration of the multilayer network, and we call it a feasible link configuration if $\Prob[X_\ell = x_\ell, \forall \ell \in E] > 0$. 

\prettyref{fig:multilayer-network} shows a three-layer network with $p_\ell=1$ for all $\ell$ (i.e. no link is thinned) and $K=1$.  Inactive links on each layer and in the merged network at the bottom are represented by dashed lines. The bottom graph is $\cG=\bigcup_{m=1}^3 G^{(m)}$. 
In this network, node $v_1$ belongs to two layers, $v_2$ belongs to three layers and these nodes have two layers in common,  and $v_3$ belongs to one layer.


\section{Hardness Results}\label{sec:hardness}

In this section, we show that it is very hard to compute the probability of given link configurations in arbitrary multilayer networks. We show in  \prettyref{subsec:hardness-single} that one source of hardness is the generality of the underlying connectivity network $G$. On the other hand, we show in \prettyref{subsec:hardness-multi} that hardness may arise from the multilayer nature of the problem, even when the underlying network $G$ has a simple structure such as a clique, which makes the single layer problem easy. We assume $K=1$ throughout this section.

\subsection{Hardness for Single Layer General Graphs}\label{subsec:hardness-single}

In this section, we show that it is hard to compute link configuration probabilities for general underlying connectivity network $G$ even when there is only one layer , i.e. $M=1$. The proof uses a reduction from the \textsc{\#Independent Set} problem. Recall that an independent set is a set of vertices in a graph no pair of which are adjacent.

\begin{definition}[\textsc{\#Independent Set}] 
Given a graph $G$, count the number of independent sets in $G$.
\end{definition}

Corollary 4.2 of \cite{vadhan2001complexity} shows that many special cases of \textsc{\#Independent Set} is \#P-complete, and hence the following

\begin{lemma}[\cite{vadhan2001complexity}]\label{lemma:sharp_P}
\#Independent Set is \#P-hard.
\end{lemma}

Recall that \#P is the class of counting problems that correspond to the decision problems in the class NP; while a decision problem asks whether there exists a solution, the corresponding counting problem asks how many solutions there are. Note that \#P-complete problems are NP-hard.
The following lemma of \cite{roth1996hardness} shows that \textsc{\#Independent set} is hard even to approximate.

\begin{lemma}[Lemma A.3 of \cite{roth1996hardness}]\label{lemma:inapproximable}
For any $\epsilon>0$, approximating the number of independent sets of a graph on $n$ vertices within $2^{n^{1-\epsilon}}$ is NP-hard.
\end{lemma}

Now we show that it is hard to compute the probability of the  configuration where no link is active, even when there is no link thinning and the probability for a node to be active is $1/2$ for all nodes. It follows that the general case is also hard.

\begin{proposition}
Suppose $M=1$, $p_\ell=1$ for all $\ell\in E$, and $q_i = 1/2$ for all $i\in V$. 
It is \#P-hard to compute the probability that the network is in the configuration with no active link. It is NP-hard to approximate this probability within a multiplicative factor of $2^{n^{1-\epsilon}}$ for any $\epsilon>0$, where $n=|V|$ is the number of nodes.
\end{proposition}

\begin{proof}
Given a graph $G=(V,E)$ and any node configuration $a:V\to \{0,1\}$, let $s(a) = a^{-1}(1)$ be the set of {active} nodes. Let $A$ denote the set of node configurations that results in an empty link set $\cE$. Note that $a\in A$ if and only if $s(a)$ is an independent set of $G$. Since $q=1/2$, all node configurations are equally likely and there are $2^n$ of them, so $|A| = 2^{n}\Prob[A]$.
Thus counting the number of independent sets in $G$ is equivalent to computing $\Prob[A]$, as one can be easily obtained from the other through rescaling. Since it is  \#P-hard to compute $|A|$ by Lemma \ref{lemma:sharp_P}, it is \#P-hard to compute $\Prob[A]$. By Lemma \ref{lemma:inapproximable}, it is NP-hard to approximate $\Prob[A]$ within a multiplicative factor of $2^{n^{1-\epsilon}}$. 
\end{proof}

\subsection{Hardness for Multilayer Cliques}\label{subsec:hardness-multi}

In this section, we show that hardness arises in yet another dimension. Consider the case that the underlying network $G$ is a clique. In this case, it is trivial to compute link configuration probabilities for a single layer\footnote{For a single layer, there is no active link if and only if at most one node is active; for configurations with at least one active link, a node is active if and only if it is the end point of an active link.}, but for a large number of layers, the problem becomes hard. In fact, it is hard even to test the feasibility of a configuration, which is a simpler problem, since a configuration is feasible if and only if its probability is nonzero. Consider the \textsc{Multilayer Clique Configuration} (MCC) problem defined below. 

\begin{definition}[\textsc{Multilayer Clique Configuration}]
Given an $M$-layer network with $G$ being a clique and a link configuration $x$, decide whether $x$ is feasible. Denote an instance by $(x,M)$.
\end{definition} 

Given any link configuration $x$, let $G(x)$ be the subgraph induced by the {active} links in $x$, i.e. $G(x) = (V, E(x))$ where $E(x) = \{\ell\in E: x_\ell = 1\}$. We have the following feasibility test.

\begin{lemma}\label{lemma:clique}
Suppose the underlying network $G$ is a clique. A link configuration $x$  is feasible if and only if the induced subgraph $G(x)$ is covered by at most $M$ cliques. 
\end{lemma}

\begin{proof}
If $x$ is feasible, then $x = \bigvee_{m=1}^M x^{(m)}$, where $x^{(m)}$ is a feasible link configuration of the $m$-th layer, and $\vee$ is component-wise  maximum. Since $G$ is a clique, so is the subgraph $G(x^{(m)})$ induced by $x^{(m)}$, if it is not empty. Thus $G(x) = \bigcup_{m=1}^M G(x^{(m)})$ is covered by at most $M$ cliques. 

For the reverse direction, suppose $G(x)$ can be covered by $M'\leq M$ cliques $C_1, \dots, C_{M'}$. On the $m$-th layer $G^{(m)}$, set a node to be {active} if and only if it is in $C_m$, which is a node configuration with positive probability. The resulting link configuration of the $M$-layer clique is exactly $x$, so $x$ is feasible.
\end{proof}

The above proof shows that any instance of MCC with $M=1$ is easy; the configuration $x$ is feasible if and only if the graph $G(x)$ induced by the {active} links in $x$ is itself a clique. For $M\geq |E|=n(n-1)/2$, where $n=|V|$ is the number of nodes, $x$ is always feasible, since $G(x)$ can be covered by $|E|$ links, which are cliques of size 2.  For the general case, however, we now show it is NP-complete by reduction from the \textsc{Clique Edge Cover} (CEC) problem, which is known to be NP-complete. Recall

\begin{definition}[\textsc{Clique Edge Cover}]
Given a graph $G$ and an integer $k$, decide whether all edges of $G$ can be covered by at most $k$ cliques in $G$. Denote an instance by $(G,k)$.
\end{definition}

\begin{lemma}[Theorem 8.1 of \cite{orlin1977contentment}]\label{lemma:CEC}
CEC is NP-complete.
\end{lemma}

We have the following

\begin{proposition}
MCC is NP-complete.
\end{proposition}

\begin{proof}
MCC is clearly in NP. We show that it is NP-hard by reduction from CEC.
Fix an instance $(G,k)$ of CEC. Consider the clique $C$ that has the same node set as $G$. Let $x$ be the link configuration of $C$ such that $x_\ell=1$ if and only if $\ell$ is a link in $G$, i.e. $G$ is the subgraph of $C$ induced by $x$. Now we obtain an instance $(x, k)$ of MCC. The conclusion then follows from \prettyref{lemma:clique} and  \prettyref{lemma:CEC}.
\end{proof}


\section{Exact Results}\label{sec:exact}

In this section, we provide recursions for computing link configuration probabilities. In the case of trees and lines, the recursions can be turned into a pseudopolynomial algorithm. In \prettyref{subsec:diff-recursions}, we discuss two different ways of doing recursions. We then consider line and tree networks in Sections \ref{ssec:line} and \ref{ssec:tree}, respectively.

\subsection{Two Different Ways for Recursion}\label{subsec:diff-recursions}

There are two natural ways to obtain recursions for link configuration probabilities of a multilayer network. One is to do recursion on the number of layers and the other on the number of nodes. We briefly discuss the former in the present section and leave the latter for Sections \prettyref{ssec:line} and \prettyref{ssec:tree}. We restrict our discussion to the case $K=1$ in this section.

Consider an $M$-layer network with a general underlying graph $G=(V,E)$. For $K=1$, the merged network is $\cG = \bigcup_{m=1}^M G^{(m)}$. Now considered a network $\cG^{(k)}$ obtained by merging only the first $k$ layers, i.e. ${\cG}^{(k)}: =\bigcup_{m=1}^k G^{(m)}$, for $k=1,2,\dots, M$.  Recall that a link is active in the $\cG^{(k)}$ if it is active in at least one layer $1,\ldots,k$.
Let  $Q_k(x)$ be the probability of the link configuration  $x\in \{0,1\}^{|E|}$ in $\cG^{(k)}$.
We have the following recursion,
\begin{equation}
\label{recur-layer}
Q_{k+1}(x)= \sum_{ y\in Y(x)}  Q_k(y) Q_1(x-y)
\end{equation}
for all $x\in \{0,1\}^{|E|}$ and $k=1\ldots, M$,
where $Y(x):=\{ y\in \{0,1\}^{|E|} : y \leq  x \}$  is the set of vectors in $\{0,1\}^{|E|}$ component-wise smaller than or equal to vector $x\in  \{0,1\}^{|E|}$.  

For instance, if $|E|=3$ and $x=(0,1,1)$ then 
\[
Y(x)=\{(0,1,1), (0,0,1), (0,1,0), (0,0,0)\}.
\]
Note that for any vectors $x, y \in \{0,1\}^{|E|}$ such that $y\leq x$, the vector $x-y$ is  also a vector in $\{0,1\}^{|E|}$.

Recursion  \eqref{recur-layer}  shows that if one know $Q_1(x)$ for all $x\in \{0,1\}^{|E|}$ then one can determine the probability configuration of any $m$-layer graph. However, as we have seen in \prettyref{subsec:hardness-single}, it is not easy to compute even $Q_1(x)$ for general $G$. Moreover, the number of terms in the summation in \eqref{recur-layer} is exponential in the graph size. Thus \eqref{recur-layer} is feasible only for very small graphs. 

Note that \eqref{recur-layer} still requires exponential time even when the underlying graph has a simpler structure such as a tree. As we will see in the next two sections, recursions on the number of nodes lead to computationally more efficient algorithms for tree networks, although we do not know how to do it for general graphs.

\label{ssec:recursion-layers}

\subsection{Links in Series}
\label{ssec:line}
Throughout this section, we assume that $q_v=q$ for all nodes $v$, $p_\ell=1$ for all  links $\ell$, and $K=1$. Section \ref{ssec:tree} presents a general algorithm for trees that allows for arbitrary $q_v$, $p_\ell$ and $K$.

Consider the graph $G_n=(V_n,E_n)$ defined by $V_n=\{1,2,\dots,n+1\}$ and $E_n=\{ e_j\}_{j=1}^{n}$, where $e_j=(j,{j+1})$ for $j=1,\ldots,n$. 
In other words, $G_n$ is composed of $n+1$ nodes and $n$ links in series, $e_1,e_2,\ldots,e_{n}$; see  \prettyref{fig:line-network}.

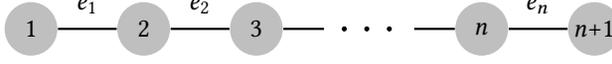
\begin{figure}[h]
\begin{center}
\captionsetup{skip=5pt}
\begin{tikzpicture}[scale=1.5, auto,swap]
    \pgfmathsetmacro{\r}{7.9}
    \foreach \x in {1,2,3}
        \node[vertex] (\x) at (\x,0) {$\x$};    
    \node[vertex] (4) at (5,0) {$n$}; 
    \node[vertex] (5) at (6,0) {$n$+$1$}; 
            
    \foreach \x in {0,0.2,0.4}
	\draw[fill=black] (3.78+\x,0) circle (0.5pt);

    \coordinate (a) at (3.6,0);
    \coordinate (b) at (4.4,0); 
	
	\graph[edges=arc]{
	(1)--(2)--(3)--(a); (b)--(4)--(5);
	};
   \foreach \x in {1,2}
        \node  at (\x+.5,.2) {$e_\x$};    
  \node  at (5.5,.2) {$e_{n}$};  
\end{tikzpicture}
\caption{Line network $G_n$.}
\label{fig:line-network}
\end{center}
\end{figure}

We are interested in calculating  $Q_{G_n, x^{(n)}}$, the probability that links $e_1,\ldots,e_{n}$ are in state $x^{(n)}:=(x_1,\ldots,x_{n})\in \{0,1\}^n$.
Link $e_j$ is in state $1$ (resp. state $0$), denoted by $x_j=1$ (resp. $x_j=0$), if it is active (resp. inactive), namely, if
nodes $j$ and ${j+1}$ belong to at least one common layer (resp. do not have any layer in common). Let $Q_{G_n, x^{(n)}}(m)$  denote the probability that $G_n$ is in state $x^{(n)}$ given that node $n+1$ belongs to $m$ layers.
By symmetry, $Q_{G_n, x^{(n)}}(m)$ is also the probability of state $x^{(n)}$ given that node $n+1$ belongs to an arbitrarily fixed set of $m$ layers.
Since no confusion occurs, henceforth we will drop the subscript  $G_{n}$ in both $Q_{G_n, x^{(m)}}$ and $Q_{G_n, x^{(n)}}(m)$ .
With a slight abuse of notation, $0^{(n)}=(0,\ldots,0)$ and $1^{(n)}=(1,\ldots,1)$, where each vector has $n$ entries.

The following recursion holds for $Q_{x^{(n)}}(m)$, $n\geq 2$,
\begin{align}
Q_{x^{(n)}}(m)&=\bar x_{n} \bar q^m \sum_{i=0}^{M-m} {M-m\choose i} q^i \bar q^{M-m-i}Q_{x^{(n-1)}}(i) \nonumber \\
&\quad +x_{n}  \sum_{j=1}^m {m\choose j} q^j \bar q^{m-j}  \sum_{i=0}^{M-m} {M-m\choose i} q^i \bar q^{M-m-i}  Q_{x^{(n-1)}}(i+j), \label{eq:Qxnm}
\end{align}
with $\bar q = 1- q$, $\bar x_n = 1- x_n$, and $Q_{x^{(0)}}(\cdot)=1$ by convention. In particular, 
\begin{equation}\label{eq:Qx2m}
Q_{x^{(1)}}(m)=x_1 (1-\bar q^m)+\bar x_1\bar q^m.
\end{equation}
The first term in the r.h.s. of \eqref{eq:Qxnm} accounts for the fact that if link $e_{n}=(n,n+1)$ is in state $x_{n}=0$
then node $n$ cannot belong to the same layer as node $n+1$ (this occurs with probability $(1-q)^m$ as node $n+1$ belongs to $m$ layers) 
but otherwise can belong to any of the  $M-m$ remaining layers,  
while the second term accounts for the fact that if link $e_{n}$ is in state $x_{n}=1$ then
node $n$ needs to share at least one layer with node $n+1$ but otherwise can belong to any other layer(s).\\

We first calculate $Q_{0^{(n)}}(m)$, the probability that links $e_1,\ldots,e_{n}$ are all inactive given that node $n+1$ belongs to $m$ layer.  This probability will turn out to be a key ingredient in  the calculation of $Q_{x^{(n)}}$.

\begin{proposition}[Calculation of $Q_{0^{(n)}}(m)$] 
\label{line-all-empty-m_layers}
For any integer  $m=0,1,\ldots,M$, $Q_{0^{(1)}}(m) = \bar q^m$, and 
$n\geq 2$,
\begin{align}
Q_{0^{(n)}}(m)&=
\bar q^m \left[1-(n-2)q^2 + P_{n-1}(q)\right]^m   \left[1-(n-1)q^2 + P_{n}(q)\right]^{M-m}, \label{eq:Q0nm}
\end{align}
where $P_1,\ldots,P_{n}$ are polynomials in the variable $q$, recursively defined by
\begin{equation}
P_{k}(q)=\bar q P_{k-1}(q) + q \bar q P_{k-2}(q)+ q^3 + (k-3) q^4, \label{def-Pk}
\end{equation}
for  $k=3,\ldots,n$, with
\begin{equation}
P_1(q)=P_2(q)\equiv 0. \label{P1-P2}
\end{equation}
\end{proposition}

\begin{proof}
For $n=1$, the conclusion follows from \eqref{eq:Qx2m}.
For $n=2$, setting $x^{(2)} = 0^{(2)} $ in \prettyref{eq:Qxnm} yields
\[
Q_{0^{(2)}}(m)= \bar q^m \sum_{i=0}^{M-m} {M-m\choose i} q^i \bar q^{M-m-i}Q_{0^{(1)}}(i).
\]
Using $Q_{0^{(1)}}(i) = \bar q^i$, we obtain $Q_{0^{(3)}}(m)= \bar q^m (1-q^2)^{M-m}$, in agreement with \eqref{eq:Q0nm}.

We now use induction to complete the proof. Assume that (\ref{eq:Q0nm}) holds for $n=2,\ldots,n^\prime-1$ with $n^\prime\geq 3$, and let us show that it still holds for $n=n^\prime$.
Setting $x_1, \dots,x_{n'-1}$ to zero in \eqref{eq:Qxnm} and using the induction hypothesis and the binomial expansion yields
\begin{align*}
Q_{0^{(n^\prime)}}(m)&= \bar q^m \left[1-(n^\prime-2)q^2+P_{n^\prime-1}(q)\right]^m\\
&\quad \times  \Bigl[1-(n^\prime-1)q^2  +\bar q P_{n^\prime-1}(q) + q \bar q  P_{n^\prime-2}(q) +q^3+(n^\prime-3)q^4 \Bigr]^{M-m}\\
&= \bar q^m \left[1-(n^\prime-2)q^2+P_{n^\prime-1}(q)\right]^m  \left[1-(n^\prime-1) q^2 + P_{n^\prime}(q)\right]^{M-m},
\end{align*}
where the latter equality comes from the definition of $P_{n^\prime}(q)$. This completes the proof.
\end{proof}

\begin{corollary}[Calculation of $Q_{0^{(n)}}$]
\label{corr:line-all-empty}
For any integer $n\geq 1$,
\[
Q_{0^{(n)}}=\left[1-nq^2 +P_{n+1}(q)\right]^M.
\]
\end{corollary}
The proof is straightforward by using Proposition \ref{line-all-empty-m_layers} together with the identify 
\begin{equation}
Q_{0^{(n)}}=\sum_{m=0}^M {M\choose m} q^m (1-q)^{M-m} Q_{0^{(n)}}(m).
\label{Q0k}
\end{equation}

We are now in position to find $Q_{x^{(n)}}$, the probability that links $e_1,\ldots,e_{n}$ are in state $(x_1,\ldots,x_{n})$. This result will be an easy consequence 
(see Proposition \ref{prop:line})
of the next proposition that determines $Q_{x^{(n)}_k}(m)$, the probability that links $e_k,\ldots,e_{n}$ are in state $(x_k,\ldots,x_{n})$ given that node $v_k$ belongs to $m$ layers, 
for $k=n,\ldots,1$.

\begin{proposition}[Calculation of $Q_{x^{(n)}}(m)$]
\label{prop:line-m_layers}
For any integer $n\geq 1$, $m=0,1,\ldots,M$,
\begin{equation}
Q_{x^{(n)}}(m)= \sum_{j=0}^{n} h_{j} Q_{0^{(n-j)}} (m)  \prod_{l=j+1}^{n}(\bar x_{l} -x_{l}),  \label{qxkk-m}
\end{equation}
where $Q_{0^{(0)}}(\cdot)\equiv 1$, $\{Q_{0^{(j)}}(m)\}_{j=1}^{n-1}$ is given in (\ref{eq:Q0nm}),  
and $\{h_j\}_{j=0}^{n}$ are mappings depending on $x^{(n)}$ and $q$ recursively defined by
\begin{equation}
h_{j} = x_{j}\sum_{r=0}^{j-1} \left[1-(j-1-r)q^2 +P_{j-r}(q)\right]^M  h_r   \prod_{l=r+1}^{j-1}(\bar x_{l} -x_{l}),
 \label{def-hk-m}
\end{equation}
for $j=2,\ldots,n$, with
\begin{equation}
h_1=x_1 \hbox{ and } h_{0}=1  \label{def-hk_hk+1}.
\end{equation}
\end{proposition}

\begin{proof}
We use the convention $\prod_{l=1}^0 \cdot =1$. Letting $n=1$ in (\ref{qxkk-m}) and using (\ref{eq:Q0nm}), (\ref{def-hk_hk+1}) and $Q_{0^{(0)}}(\cdot)\equiv 0$ yield
\begin{eqnarray*}
Q_{x^{(1)}}(m)&=&h_{0} Q_{0^{(1)}}(m) \left(\bar x_1 - x_1\right) +h_1 Q_{0^{(0)}}(m)\\
&=& x_1\left(1-\bar q^m\right) +\bar x_1 \bar q^m,
\end{eqnarray*}
which is true by (\ref{eq:Qx2m}).
Assume that (\ref{qxkk-m}) is true for $n=1,\ldots,n^\prime-1$. We show that it is still true for $n=n^\prime$.

From the induction assumption (\ref{qxkk-m}) with $n=n^\prime-1$, relation (\ref{Qxkm}) and Lemma \ref{lem2} both given in Appendix \ref{sec:app}, we get 
\begin{align}
Q_{x^{(n^\prime)}}(m)&=\bar x_{n^\prime} \bar q^m \sum_{j=0}^{n^\prime-1} h_j F(Q_{0^{n^\prime-1-j}},m) \prod_{l=j}^{n^\prime-1} (\bar x_l-x_l) + x_{n^\prime} \sum_{j=0}^{n^\prime-1} h_j G(Q_{0^{n^\prime-1-j}},m) \prod_{l=j}^{n^\prime-1} (\bar x_l-x_l)\nonumber\\
&= \sum_{j=0}^{n^\prime-1} h_j Q_{0^{n^\prime-j}}(m) \prod_{l=j+1}^{n^\prime} (\bar x_l - x_l) +  \sum_{j=0}^{n^\prime-1} \left[1-(m^\prime-1-j)q^2+P_{n^\prime-j}\right]^M  h_j \prod_{l=j+1}^{n^\prime-1}(\bar x_l - x_l)\nonumber\\
&= \sum_{j=0}^{n^\prime-1} h_j Q_{0^{n^\prime-j}}(m) \prod_{l=j+1}^{n^\prime} (\bar x_l - x_l) + h_{n^\prime}\label{proof-induc}\\
&=  \sum_{j=0}^{n^\prime} h_j Q_{0^{n^\prime-j}}(m) \prod_{l=j+1}^{n^\prime} (\bar x_l - x_l),\nonumber
\end{align}
where (\ref{proof-induc}) follows from the definition of $h_{n^\prime}$ given in (\ref{def-hk-m}). This concludes the induction step.
\end{proof}

\begin{proposition}[Calculation of $Q_{x^{(n)}}$]\hfill
\label{prop:line}

For any integer $n\geq 1$,
\begin{equation}
Q_{x^{(n)}}= \sum_{j=0}^{n} h_{j} Q_{0^{(n-j)}}  \prod_{l=j+1}^{n}(\bar x_{l} -x_{l}),
\label{qxkm}
\end{equation}
where $Q_ {0^{(0)}}=1$, and $Q_ {0^{(1)}}\ldots, Q_ {0^{(n)}}$ are given in Corollary \ref{corr:line-all-empty} and $h_0,\ldots,h_{n}$ are given  in (\ref{def-hk-m})-(\ref{def-hk_hk+1}).
\end{proposition}

\begin{proof}
The proof  follows from the identity
\[
Q_{x^{(n)}}=\sum_{m=0}^M {M\choose m} q^m \bar q^{M-m}Q_{x^{(n)}}(m)
\]
together with Proposition \ref{prop:line-m_layers} and (\ref{Q0k}).
\end{proof}

We conclude this section by calculating the expected value and the probability generating function (pgf) of the size of the connected component a node belongs to, and the pgf of the number of active links.

For a path of length $n$ in \prettyref{fig:line-network}, let  $C_{n,i}\in\{1,\ldots,n+1\}$ denote the random variable for the size of the connected component that node $i$ belongs to, $i=1,..., n+1$. Note that $C_{n,i}$ can be rewritten as
\[
C_{n,i} = C'_{n,i} + C''_{n,i} - 1,
\]
where $C'_{n,i}$ and $C''_{n,i}$ count the number of nodes to the left and right of node $i$ (including node $i$) in the same connected component, respectively. The subtraction by 1 accounts for the double counting of node $i$.

Let $\hat C_{n,i}(z)=\E[z^{C_{n,i}}]$ denote the pgf for the distribution of $C_{n,i}$, and 
\[
\hat C_{n}(z;m) :=\E[z^{C_{n,n+1}}\,|\, \hbox{node $n+1$ belongs to } m \hbox{ layers}  ]
\]
the conditional pgf of $C_{n,n+1}$. Note that given the number of layers that node $i$ belongs to, $C'_{n,i}$ and $C''_{n,i}$ are conditionally independent and have the same conditional distributions as $C_{i-1,i}$ and $C_{n-i+1, n-i+2}$, respectively. Therefore,
\begin{equation}
\hat C_{n,i}(z)=\frac{1}{z}\sum_{m=0}^M {M\choose m}q^m \bar q^{M-m}\hat C_{i-1}(z;m)  \hat C_{n-i+1}(z;m).
\label{Cniz}
 \end{equation}
We are left with finding $\hat C_{n}(z;m)$. Note that $\hat C_0(z;m) = z$ and $\hat C_1(z;m) = \bar q^m z + (1-\bar q^m) z^2$. Define the vector $y^{(i)}$ of size $i\geq 2$ by  $y^{(i)}=(0,1,\ldots,1)$. For $n\geq 2$,  \prettyref{prop:line-m_layers} yields
\begin{align}
\hat C_{n}(z;m)&= Q_{0^{(1)}}(m) z + \sum_{i=2}^{n}z^ i Q_{y^{(i)}}(m) + z^{n+1}Q _{1^{(n)}}(m)\nonumber\\
&=  Q_{0^{(1)}}(m)  z + \sum_{i=2}^n z^i(-1)^{i-1} \hat h_{0,0}Q_{0^{(i)}}(m) +\sum_{i=2}^{n} z^i \sum_{j=1}^{i}  (-1)^{i-j} \hat h_{0,j} Q_{0^{(i-j)}}(m) \nonumber\\
& +z^{n+1}\sum_{j=0}^n  (-1)^{n-j} \hat h_{1,j} Q_{0^{(n-j)}}(m),
\label{Ck1}
\end{align}
where $Q_{0^{(1)}}(m),\ldots, Q_{0^{(n)}}(m)$  are given in Proposition  \ref{line-all-empty-m_layers} 
(with $Q_{0^{(0)}}(\cdot)=1$), $P_1(q),\ldots,P_n(q)$ are recursively defined in (\ref{def-Pk}) and (\ref{P1-P2}), and for $b \in \{0,1\}$, $\hat h_{b,0} = 1$, $\hat h_{b,1} = b$, and
\[
\hat h_{b,j}=(-1)^{j-b} \left(1-(j-1)q^2 +P_{j}(q)\right)^M +\sum_{r=1}^{j-1} (-1)^{j-r -1} \left(1-(j-1-r)q^2 +P_{j-r}(q)\right)^M  \hat h_{b,r}
\]
for $j=2,3,\dots,n$.

Let $\bar C_{n,i} =\E [C_{n,i}]$ be the expected size of the connected component node $i$ belongs to. 
From $\bar C_{n,i}=d\hat C_{n,i}(z)/dz|_{z=1}$,  (\ref{Cniz}) and  \eqref{Ck1}, we obtain  
\[
\bar C_{n,i} = \bar C_{i-1,1} + \bar C_{n-i+1,1} - 1,
\]
and
\[
\bar C_{n,1}=Q_{0^{(1)}} + \sum_{i=2}^n i(-1)^{i-1} \hat h_{0,0}Q_{0^{(i)}}+ \sum_{i=2}^{n} i \sum_{j=1}^{i}  (-1)^{i-j} \hat h_{0,j} Q_{0^{(i-j)}} +(n+1)\sum_{j=0}^n  (-1)^{n-j} \hat h_{1,j} Q_{0^{(n-j)}},
\]
where $Q_{0^{(1)}},\ldots,Q_{0^{(n)}}$  are given in Corollary  \ref{corr:line-all-empty} (and $Q_{0^{(0)}}=1$).
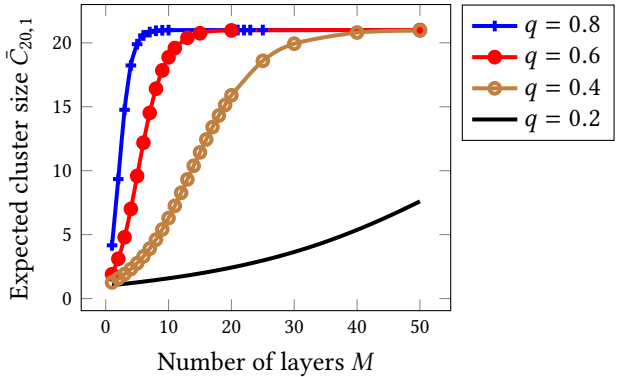
\begin{figure}[h]
\begin{center}
\begin{tikzpicture}
\begin{axis}[small, name=plot2,
   legend pos=outer north east, 
   axis on top,
   ylabel near ticks,
   xlabel near ticks,
   ylabel = {Expected cluster size $\bar C_{20,1}$},
   xlabel = {Number of layers $M$},
]
 \addplot  [smooth, mark=+, blue] table [x=a, y=b] 
 {                         a             b
                          1	4.163106496
                         2  	9.343847708
                         3	14.76208392
                         4	18.22997835
                         5	19.88956040
                         6	20.57702610
                         7	20.84293114
                         8 	20.94243512
                         9      20.97905350
                          10      20.99240937
                        21 20.99999989
                        22 20.99999996
                        23 20.99999997
                        25  21                                               
                        50  21 
};
\addlegendentry{$q=0.8$}

\addplot  [smooth, mark=*, red] table [x=a, y=b] 
{                         a             b
                         1  1.899967087
                         2  3.111115783
                         3  4.794176613
                         4  7.008073694
                         5  9.583883131
                         6  12.18923908
                         7  14.51896785
                         8  16.41116077
                         9 17.84296556
                        10 18.87304688
                        11 19.58832458
                        13  20.39568104
                        15   20.74682213
                        20   20.97226030
                        50  21
};
\addlegendentry{$q=0.6$}

\addplot  [smooth, mark=o , brown]table [x=a, y=b] 
 {                         a             b
                             1  1.266666661
                         2  1.567976967
                         3  1.911928042
                         4  2.307572349
                         5  2.765076770
                         6  3.295271768
                         7  3.908352888
                         8  4.611752367
                         9  5.407749714
                        10 6.291695770
                        11    7.251540856
                        12    8.268835507
                        13    9.320851765
                        14    10.38319772
                        15    11.43230741
                        16    12.44738107
                        17    13.41159472
                        18    14.31259855
                        19    15.14244435
                        20    15.89712711
	                25    18.59211217 
                        30    19.93338816
                        40    20.80585003
                        50    20.96573203
};
\addlegendentry{$q=0.4$}

\addplot  [color=black, mark=] table [x=a, y=b] 
 {                         a             b
                        1 1.050000000
                         2 1.101693549
                         3  1.155158493
                         4  1.210475376
                         5 1.267727580
                         6  1.327001488
                         7 1.388386630
                         8  1.451975847
                         9 1.517865455
                        10 1.586155412
                        11    1.656949484
                        12    1.730355411
                        13    1.806485072
                        14    1.885454616
                        15    1.967384605
                        16    2.052400075
                        17    2.140630565
                        18    2.232210057
                        19    2.327276803
                        20    2.425973040
                        21    2.528444526
                        22    2.634839908
                        23    2.745309882
                        24    2.860006148
                        25    2.979080140
                        26    3.102681547
                        27    3.230956597
                        28    3.364046216
                        29    3.502083986
                        30    3.645194023
                        31    3.793488766
                        32    3.947066780
                        33    4.106010559
                        34    4.270384471
                        35    4.440232801
                        36    4.615578014
                        37    4.796419250
                        38    4.982731090
                        39    5.174462614
                        40    5.371536813
                        41    5.573850300
                        42    5.781273378
                        43    5.993650432
                        44    6.210800663
                        45    6.432519085
                        46    6.658577841
                        47    6.888727739
                        48    7.122700008
                        49    7.360208244
                        50    7.600950500
                        };
\addlegendentry{$q=0.2$}

\end{axis}
\end{tikzpicture}
\caption{Expected size $\bar C_{20,1}$ of the cluster containing node $1$ in the line network of  \prettyref{fig:line-network} as a function of the number of layers $M$,  for $n=20$ and $q\in\{0.2,0.4,0.6,0.8\}$. }                      
\label{fig-line-expected-cluster-size}
\end{center}
\end{figure}

\begin{figure}[h]
\begin{center}
\begin{tikzpicture}
\begin{axis}[small, name=plot2,
   legend  pos=outer north east, 
   axis on top,
   ylabel near ticks,
   xlabel near ticks,
   ylabel = {Expected cluster size $\bar C_{20,1}$},
   xlabel = {Number of layers $M$},
]

\addplot  [color=blue] table [x=a, y=b] 
 {                         a             b
                         2  2.107097221
                         3  1.547096933
                         4  1.360451006
                         5 1.267727581
                         9 1.130742721
                        13 1.086145217
                        17 1.064163728
                        21 1.051098355
                        25 1.042445588
                        29 1.036295279
                        33 1.031699934
                        37 1.028136471
                        41 1.025292633
                        45 1.022970547
                        49 1.021038759
};
\addlegendentry{$q=M^{-1}$}

\addplot  table [x=a, y=b] 
 {                         a             b
                         2 5.349414064
                         6 3.466524162
                        10 3.212118325
                        14 3.104113450
                        18 3.042422692
                        22 3.001755404
                        26 2.972567387
                        30 2.950403268
                        34 2.932883732
                        38 2.918613627
                        42 2.906716476
                        46 2.896611634
                        50 2.887897896
};
\addlegendentry{$q=M^{-1/2}$}

\addplot  table [x=a, y=b] 
 {                         a             b
                         2 8.989749997
                         6 9.047208889
                        10 10.41862901
                        14 11.63260947
                        18 12.66264506
                        22 13.53622581
                        26 14.28168691
                        30 14.92248645
                        34 15.47729286
                        38 15.96088664
                        42 16.38502382
                        46 16.75912703
                        50 17.09081293
};
\addlegendentry{$q=M^{-1/3}$}
\end{axis}
\end{tikzpicture}
\caption{Expected size $\bar C_{20,1}$ of the cluster containing node $1$ in the line network of  \prettyref{fig:line-network} as a function of the number of layers $M$,  for  $n=20$ and $q\in\{M^{-1},M^{-1/2},M^{-1/3}\}$.}                     
\label{fig-line-scalingM}
\end{center}
\end{figure}
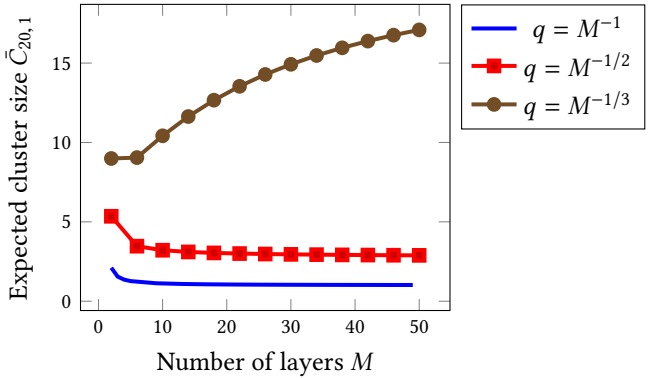

Let $L_n$ denote the number of active links in a line of length $n\geq 1$ and $\bar L_n = \E[L_n]$. Note that $\bar L_1 = 1-(1-q^2)^m$ and $\bar L_n = n \bar L_1$. Let $\hat L_n(z) = \E[z^{L_n}]$ be the pgf of $L_n$.  Note that $\hat L_n(z)$ is expressed as
\[ 
\hat L_n(z) = \sum_{m=0}^M \binom{M}{m} q^m \bar q^{M-m}\hat L_n^{(m)}(z), \quad n\geq 1, 
\]
where $\hat L_n^{(m)}(z) = \E[L_n\mid \text{node $n+1$ belongs to $m$ layers}]$ is the conditional pgf 
for $m=0,\ldots , M$.  The following recursion holds for $L_n^{(m)}(z)$, $n\geq 2$, 
\begin{equation}\label{eq:Qxnm}
\hat L_n^{(m)}(z) =
\bar q^m \sum_{i=0}^{M-m} {M-m\choose i} q^i \bar q^{M-m-i}\hat L_{n-1}^{(i)}(z)+ z \sum_{j=1}^m {m\choose j} q^j \bar q^{m-j}  \sum_{i=0}^{M-m} {M-m\choose i} q^i \bar q^{M-m-i} \hat L_{n-1}^{(i+j)}(z) ,  
\end{equation}
with  $\hat L_1^{(m)}(z) =\bar q^m + z (1- \bar q^m )$ for $m=0,1,\ldots,M$, and $\sum_{j=1}^0 \cdot = 0$ by convention.  Note that $\bar L_n=\sum_{m=0}^M {M\choose m} q^m \bar q^{M-m} d\hat L^{(m)}_n(z)/dz|_{z=1}$, but it is much easier to use the formula $\bar L_n = n[1-(1-q^2)^m]$.

Figures \ref{fig-line-expected-cluster-size} and \ref{fig-line-expected-number-active-links} display the mappings $M\to \bar C_{n,1}$ and  $M\to \bar L_n$ for $n=20$, respectively, when $q\in \{0.2, 0.4, 0.6, 0.8 \}$.
It shows the impact of having a finite number of layers on these metrics.
Figures \ref{fig-line-scalingM} and \ref{fig-line-expected-number-active-links-scaling} investigate the behavior of these mappings  when $q\in\{M^{-1}, M^{-1/2}, M^{-1/3}\}$.  These plots show that both $\bar C_{n,1}$ and $\bar L_n$ scale
with $M$ as $q=1/\sqrt{M}$. This result is rooted in the result that the limit of $1-(1-q^2)^M$ -- the probability that a link is active -- is non-zero as $M\uparrow \infty$ when $q=1/\sqrt{M}$ (this limit is  $1-e^{-1}$).  The asymptotic behavior of a multilayer network as $M\to \infty$ is investigated in depth  in Section \ref{sec:asymp}.

\begin{figure}[h]
\begin{center}
\begin{tikzpicture}
\begin{axis}[small, name=plot2,
   legend  pos=outer north east, 
   axis on top,
   ylabel near ticks,
   xlabel near ticks,
   ylabel = {Expected number of active nodes $\bar L_{20}$},
   xlabel = {Number of layers $M$},
]
 \addplot  [smooth, mark=+, blue] table [x=a, y=b] 
 {                         a             b
                         1 12.80000000
                         2 17.40800000
                         3 19.06688000
                         4 19.66407680
                         5 19.87906763
                         6 19.95646434
                         7 19.98432718
                         8 19.99435778
                         9 19.99796879
                        10 19.99926878
};
\addlegendentry{$q=0.8$}

\addplot  [smooth, mark=*, red]  table [x=a, y=b] 
 {                         a             b
                         1  7.200000000
                         2 11.80800000
                         3  14.75712000
                         4 16.64455680
                         5 17.85251636
                         6 18.62561046
                         7 19.12039072
                         8 19.43705004
                         9  19.63971205
                        10 19.76941569
};
\addlegendentry{$q=0.6$}

\addplot   [smooth, mark=o, brown] table [x=a, y=b] 
 {                         a             b
                         1 3.200000000
                         2 5.888000000
                         3 8.145920000
                         4 10.04257280
                         5  11.63576115
                         6 12.97403937
                         7 14.09819307
                         8  15.04248218
                         9  15.83568503
                        10 16.50197543
};
\addlegendentry{$q=0.4$}

\addplot   [smooth, mark=, black] table [x=a, y=b] 
 {                         a             b
                        1  0.8000000000
                         2 1.568000000
                         3  2.305280000
                         4  3.013068799
                         5  3.692546047
                         6 4.344844201
                         7 4.971050438
                         8 5.572208420
                         9 6.149320083
                        10  6.703347281
 };
 \addlegendentry{$q=0.2$}
\end{axis}
\end{tikzpicture}
\caption{Expected number $\bar L_{20}$ of active links in a line network with $20$ links as a  function of the number of layers $M$, for $q\in\{0.2,0.4,0.6,0.8\}$.  }                    
\label{fig-line-expected-number-active-links}
\end{center}
\end{figure}
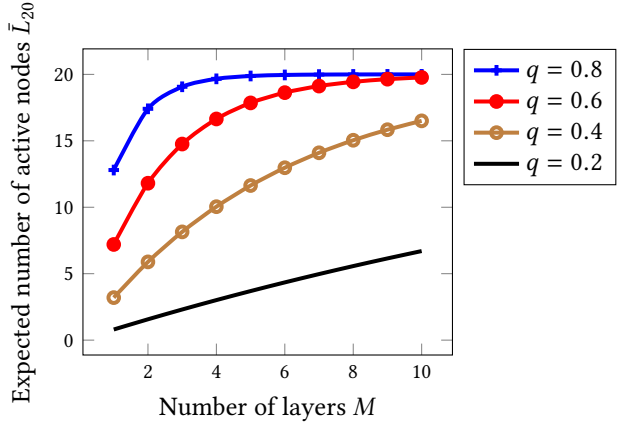

\begin{figure}[h]
\begin{center}
\begin{tikzpicture}
\begin{axis}[small, name=plot2,
   legend  pos=outer north east, 
   axis on top,
   ylabel near ticks,
   xlabel near ticks,
   ylabel = {Expected number of active nodes $\bar L_{20}$},
   xlabel = {Number of layers $M$},
]

\addplot [color=blue, mark=] table [x=a, y=b]   
 {                         a             b
                         5 3.692546048
                        10 1.912358500
                        15 1.292640206
                        20 0.9766024949
                        25  0.7848267690
                        30 0.6560364808
                        35 0.5635692880
                        40  0.4939542145
                        45  0.4396499130
                        50 0.3961049705
};
\addlegendentry{$q=M^{-1}$}

\addplot  table [x=a, y=b] 
 {                         a             b
                         5 13.44639998
                        10 13.02643123
                        15 12.89471267
                        20 12.83028153
                        25 12.79206567
                        30 12.76676970
                        35 12.74878899
                        40 12.73535119
                        45 12.72492762
                        50 12.71660636
};
\addlegendentry{$q=M^{-1/2}$}

\addplot  table [x=a, y=b] 
 {                         a             b
                         5 17.53296484
                        10 18.23285400
                        15 18.64825040
                        20 18.91829111
                        25 19.10765930
                        30 19.24758531
                        35 19.35495478
                        40 19.43972496
                        45 19.50816374
                        50 19.56442064
};
\addlegendentry{$q=M^{-1/3}$}
\end{axis}
\end{tikzpicture}
\caption{Expected number $\bar L_{20}$ of active links in a line network with $20$ links as a  function of the number of layers $M$,  for $q\in\{M^{-1},M^{-1/2},M^{-1/3}\}$.}                   
\label{fig-line-expected-number-active-links-scaling}
\end{center}
\end{figure}
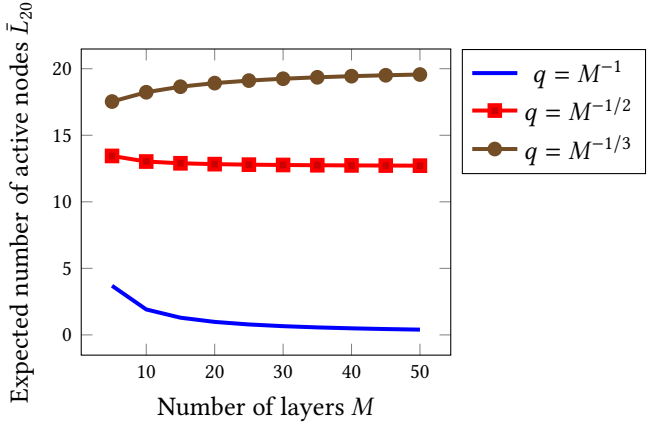

\subsection{Recursion for Multilayer Trees}
\label{ssec:tree}

In this section, we consider the case where the underlying graph is a tree $T$. We develop a recursion that provides a pseudo polynomial time algorithm to compute the probability of any link configuration. The recursion applies to all parameter settings of the model introduced in \prettyref{sec:model}.

Pick any root $r$ for $T$. For $v\in V$, let $T_v$ denote the subtree rooted at $v$. With a slight abuse of notation, let $T_v(x)$ be  the event that the links in $T_v$ are configured according to $x:E\to\{0,1\}$, i.e. $X_\ell = x_\ell$ for all $\ell \in T_v$. For $v,w\in V$, let $A_v$ be the number of layers on which $v$ is active, and $A_{vw}$ the number of layers on which both $v$ and $w$ are active. Note that $A_v$ has a binomial distribution with parameter $M$ and $q_v$, i.e.
\[
\Prob[A_v=m] = B(m;M, q_v)\triangleq \binom{M}{m} q_v^m (1- q_v)^{M-m}.
\]

For $m = 0,1,\dots, M$, define $f_{v}^{(m)}(x)$ by
\[
f_v^{(m)}(x) = \Prob[A_v=m, T_v(x)].
\] 

Let $ch(v)$ denote the set of $v$'s children. Using the conditional independence of $X_{(v,w)}$ and $T_w(x)$ for different $w\in ch(v)$, we obtain
\[
f_{v}^{(m)}(x) = \Prob[A_v=m] \prod_{w\in ch(v)} g_{vw}^{(m)}(x),
\]
where  
\[
g_{vw}^{(m)}(x) = \Prob\left[X_{(v,w)} = x_{(v,w)}, T_w(x) \bigm\vert A_v=m\right].
\]
If $v$ is a leaf node, then $ch(v) = \emptyset$ and $\prod_{w\in ch(v)} g_{vw}^{(m)}(x) = 1$ by convention.

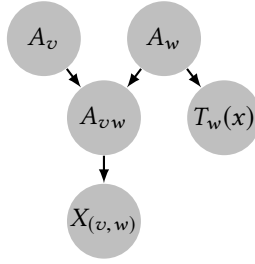
\begin{figure}
\begin{tikzpicture}[scale=.8, auto,swap]
    \node[vertex,minimum size=28pt] (v) at (0,0) {$A_v$};
    \node[vertex,minimum size=28pt] (w) at (2,0) {$A_w$};    
    \node[vertex,minimum size=28pt] (vw) at (1,-1.3) {$A_{vw}$};
    \node[vertex,minimum size=28pt] (Tw) at (3,-1.3) {$T_w(x)$};
    \node[vertex,minimum size=28pt] (x) at (1,-3) {$X_{(v,w)}$};
       
	\graph[edges=arc]{
	{(v), (w)}->(vw)->(x); (w)->(Tw);
	};

\end{tikzpicture}
\caption{Dependency structure in the calculation of $g_{vw}^{(m)}$.}
\label{fig:dependency}
\end{figure}

To compute $g_{vw}^{(m)}(x)$, we make use of the dependency structure among the variables/events shown in \prettyref{fig:dependency}. Thus
\begin{align*}
g_{vw}^{(m)}(x) &= \sum_{k=0}^M\sum_{j=0}^{\min\{m,k\}} \Prob[A_w=k, T_w(x)]  \times \Prob[A_{vw}=j\mid A_v = m, A_w = k] \\
&\hspace{21mm} \times \Prob[X_{(v,w)} = x_{(v,w)}\mid A_{vw}=j].
\end{align*}
The first factor in the summation is $f_w^{(m)}(x)$. By permutation symmetry, the second factor in the summation is given by a hypergeometric pmf, i.e.
\[
\Prob[A_{vw}=j\mid A_v = m, A_w = k] = H(j;M,m,k)\triangleq \frac{\binom{m}{j} \binom{M-m}{k-j}}{\binom{M}{k}}.
\]
For the third factor, note that given $A_{vw} = j$, the number of layers on which link $(v,w)$ is active, $W_{(v,w)}$, has a binomial distribution with parameter $j$ and $p_{(v,w)}$. Thus
\begin{align*}
\Prob[X_{(v,w)} = x_{(v,w)}\mid A_{vw}=j] = (1- x_{vw}) [1-\bar F_B(K; j, p_{(v,w)})]+ x_{vw} \bar F_B(K;j,p_{(v,w)}),
\end{align*}
where $\bar F_B(k; j, p_{(v,w)}) = \sum_{\ell\geq k} B(\ell;j,p_{(v,w)})$ is the ccdf of binomial distribution with parameter $j$ and $p_{(v,w)}$.

We can compute $\left\{f_v^{(m)}(x):m=0,1,\dots, M\right\}$ for all $v\in V$ sequentially from leaf nodes up to the root. The complexity is linear in $n=|V|$.  The probability of configuration $x$ is then
\[
\Prob[X_\ell = x_\ell, \forall \ell\in E] = \sum_{m=0}^M f_r^{(m)}(x).
\]
Since the pmfs of hypergeometric and binomial distributions can be computed in time polynomial in $M$, the above recursion can be computed in time $O(n\cdot \text{poly}(M))$. Note that this pseudopolynomial time algorithm relies on the assumption that the layers are i.i.d.. Although a similar recursion can be developed when layers are independent but have non-identical distributions, the complexity will scale as $M!$, making it feasible only for small $M$.

Figure \ref{fig:star-generalized} displays the probability that all links are active in the star shaped network represented Figure \ref{fig:star} as a function of $q$, for different number of layers.
Similarly, Figure \ref{fig:binary-tree} shows the probability that all links are active in a binary tree of depth 5  as a function of $q$, for different number of layers. These results were generated using the recursion discussed in this section with $q_v=q$ and $p_\ell=1$ for all nodes $v$ and links $l$.

\begin{figure}[h]
\centering
\begin{tikzpicture}[scale=.8, auto,swap]
    \node[vertex] (0) at (3,0) {};
    \foreach \x/\n in {1/2,2/3,3/4,4/3}
    	\foreach \y in {1,2,...,\n}
		\node[vertex] (\x\y) at (1.2*\x,-1.2*\y) {};
            
	\graph[edges=arc]{
	(0)--{(11), (21),(31),(41)}; (11)--(12); 
	(21)--(22)--(23);(31)--(32)--(33)--(34); (41)--(42)--(43);
	};

\end{tikzpicture}
\caption{A star shaped network.}
\label{fig:star}
\end{figure}
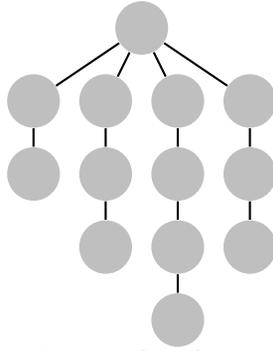

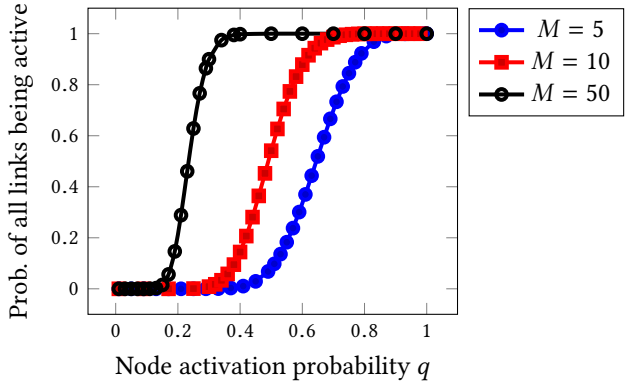
\begin{figure}[h]
\begin{center}
\begin{tikzpicture}
\begin{axis}[small, name=plot1,
   legend pos=outer north east,
   axis on top,
   ylabel near ticks,
   xlabel near ticks,
   ylabel = {Prob. of all links being active},
   xlabel = {Node activation probability $q$},
]


\addplot table [x=a, y=b]
{                            a             b
 0.0100000   0.0000000
 0.0500000   0.0000000
 0.0900000   0.0000000
 0.1300000   0.0000000
 0.1700000   0.0000001
 0.2100000   0.0000012
 0.2500000   0.0000157
 0.2900000   0.0001275
 0.3300000   0.0007263
 0.3700000   0.0031143
 0.4100000   0.0105424
 0.4500000   0.0291941
 0.4900000   0.0679615
 0.5100000   0.0978369
 0.5300000   0.1359271
 0.5500000   0.1826490
 0.5700000   0.2378584
 0.5900000   0.3007775
 0.6100000   0.3699972
 0.6300000   0.4435621
 0.6500000   0.5191290
 0.6700000   0.5941774
 0.6900000   0.6662413
 0.7100000   0.7331293
 0.7300000   0.7931021
 0.7500000   0.8449884
 0.7700000   0.8882279
 0.7900000   0.9228461
0.8300000   0.9687193
 0.8700000   0.9906075
 0.9300000   0.9994871
 0.9900000   1.0000000
  1                  1
};
\addlegendentry{$M=5$}

\addplot table [x=a, y=b]
{                            a             b
 0.0100000   0.0000000
 0.0900000   0.0000000
 0.1700000   0.0000031
 0.2500000   0.0008594
 0.3000000   0.0085048
 0.3200000   0.0176335
 0.3400000   0.0334071
 0.3600000   0.0583667
 0.3800000   0.0948042
 0.4000000   0.1441976
 0.4200000   0.2067244
 0.4400000   0.2810137
 0.4600000   0.3642283
 0.4800000   0.4524615
 0.5000000   0.5413375
 0.5200000   0.6266564
 0.5400000   0.7049329
 0.5600000   0.7737325
 0.5800000   0.8317773
 0.6000000   0.8788554
 0.6200000   0.9155984
 0.6400000   0.9432046
 0.6600000   0.9631662
 0.6800000   0.9770432
 0.7000000   0.9863007
 0.7200000   0.9922098
 0.7400000   0.9958041
 0.7600000   0.9978760
 0.7800000   0.9989995
 0.8000000   0.9995672
 0.8200000   0.9998310
 0.8400000   0.9999418
 0.8600000   0.9999829
 0.8800000   0.9999959
 0.9000000   0.9999993
 0.9200000   0.9999999
 0.9400000   1.0000000
 0.9600000   1.0000000
 0.9800000   1.0000000
};
\addlegendentry{$M=10$}

\addplot [smooth, mark=o, black] table [x=a, y=b]
{                            a             b
 0.0100000   0.0000000
 0.0300000   0.0000000
 0.0500000   0.0000000
 0.0700000   0.0000001
 0.0900000   0.0000095
 0.1100000   0.0002345
 0.1300000   0.0025553
 0.1500000   0.0151467
 0.1700000   0.0563901
 0.1900000   0.1465538
 0.2100000   0.2886876
 0.2300000   0.4605868
 0.2500000   0.6283739
 0.2700000   0.7661368
 0.2900000   0.8644988
 0.3000000   0.8997606
 0.3400000   0.9747592
 0.3800000   0.9951028
 0.4000000   0.9980426
 0.5000000   0.9999932
 0.6000000   1.0000000
 0.7000000   1.0000000
 0.8000000   1.0000000
 0.9000000   1.0000000
 1                   1
};
\addlegendentry{$M=50$}
\end{axis}

\end{tikzpicture}
\caption{Probability that all links are active in the star shaped network in  \prettyref{fig:star} as a function of the node activation probability $q$.}                      
\label{fig:star-generalized}
\end{center}
\end{figure}


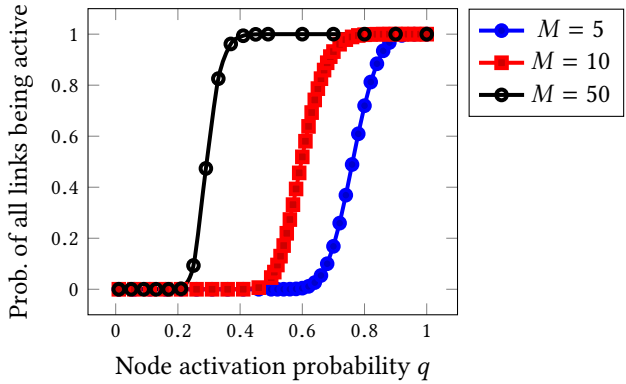
\begin{figure}[h]
\begin{center}
\begin{tikzpicture}
\begin{axis}[small, name=plot1,
   legend pos=outer north east,
   axis on top,
   ylabel near ticks,
   xlabel near ticks,
   ylabel = {Prob. of all links being active},
   xlabel = {Node activation probability $q$},
]


\addplot table [x=a, y=b]
{                            a             b
 0.0100000   0.0000000
 0.0600000   0.0000000
 0.1100000   0.0000000
 0.1600000   0.0000000
 0.2100000   0.0000000
 0.2600000   0.0000000
 0.3100000   0.0000000
 0.3600000   0.0000000
 0.4100000   0.0000000
 0.4600000   0.0000001
 0.5000000   0.0000045
 0.5200000   0.0000243
 0.5400000   0.0001119
 0.5600000   0.0004396
 0.5800000   0.0014911
 0.6000000   0.0044050
 0.6200000   0.0114340
 0.6400000   0.0262991
 0.6600000   0.0540477
 0.6800000   0.1000670
 0.7000000   0.1682948
 0.7200000   0.2592592
 0.7400000   0.3689301
 0.7600000   0.4891125
 0.7800000   0.6093553
 0.8000000   0.7195789
 0.8200000   0.8123465
 0.8400000   0.8840273
 0.8600000   0.9347085
 0.8800000   0.9672208
 0.9200000   0.9950346
 0.9600000   0.9998245
 1.0000000   1.0000000
};
\addlegendentry{$M=5$}

\addplot table [x=a, y=b]
{                            a             b
0.0100000   0.0000000
 0.0600000   0.0000000
 0.1100000   0.0000000
 0.1600000   0.0000000
 0.2100000   0.0000000
 0.2600000   0.0000000
 0.3100000   0.0000000
 0.3600000   0.0000008
 0.4100000   0.0001657
 0.4600000   0.0064071
 0.5000000   0.0458577
 0.5100000   0.0673127
 0.5200000   0.0951335
 0.5300000   0.1297897
 0.5400000   0.1713569
 0.5500000   0.2194674
 0.5600000   0.2733146
 0.5700000   0.3317130
 0.5800000   0.3932015
 0.5900000   0.4561718
 0.6000000   0.5190039
 0.6100000   0.5801875
 0.6200000   0.6384184
 0.6300000   0.6926613
 0.6400000   0.7421782
 0.6500000   0.7865269
 0.6600000   0.8255359
 0.6700000   0.8592629
 0.6800000   0.8879470
 0.6900000   0.9119575
 0.7000000   0.9317475
 0.7200000   0.9606600
 0.7400000   0.9786281
 0.7600000   0.9891237
 0.7800000   0.9948606
 0.8000000   0.9977724
 0.8200000   0.9991291
 0.8400000   0.9996999
 0.8600000   0.9999118
 0.8800000   0.9999790
 0.9000000   0.9999962
 0.9200000   0.9999995
 0.9400000   1.0000000
 0.9600000   1.0000000
 0.9800000   1.0000000
 1.0000000   1.0000000
};
\addlegendentry{$M=10$}

\addplot [smooth, mark=o, black] table [x=a, y=b]
 {                            a             b
 0.0100000   0.0000000
 0.0500000   0.0000000
 0.0900000   0.0000000
 0.1300000   0.0000000
 0.1700000   0.0000005
 0.2100000   0.0018159
 0.2500000   0.0931502
 0.2900000   0.4735923
 0.3300000   0.8255762
 0.3700000   0.9616080
 0.4100000   0.9937832
 0.4500000   0.9992442
 0.4900000   0.9999324
 0.6               0.99999
 0.7               0.99999
 0.8               0.99999
 0.9               0.99999
 1                  1
};
\addlegendentry{$M=50$}
\end{axis}
\end{tikzpicture}

\caption{Probability that all links are active in a binary tree of height
  $5$ as a function of the node activation probability $q$.}                      
\label{fig:binary-tree}
\end{center}
\end{figure}


\section{Asymptotic Results}\label{sec:asymp}

In this section, we derive the link configuration distribution in the limit as the number of layers, $M$, goes to infinity and the probabilities $\{p_\ell\}_{\ell\in E}$ and $\{q_i\}_{i\in V}$ decrease as functions of $M$. This case is especially relevant when the multilayer network arises from snapshots of a single network as in duty-cycled wireless sensor networks, where $M$ can easily become very large. Even for moderate $M$, the asymptotics may already provide good approximations as shown in \cite{Guh16}.
The main results are presented in \prettyref{subsec:asymptotic-main} with proofs given in \prettyref{subsec:asymptotic-proof}. The results are extended in \prettyref{subsec:asymptotic-non-identical} beyond the model of \prettyref{sec:model} to the case of non-identical layers.

\subsection{Main Results}\label{subsec:asymptotic-main}

Consider link $\ell=(i,j)\in E$. 
Note that its multiplicity $W_\ell$, the number of layers within which link $\ell$ is active (given in \prettyref{eq:W}) has a binomial distribution, 
\[
\Prob[W_\ell=w] = \binom{M}{w} (p_\ell q_i q_j)^w (1-p_\ell q_i q_j)^{M-w}, \quad w=1,\dotsc , M.
\]
If $M p_\ell q_i q_j$ has a finite positive limit as $M\to\infty$, a classical result shows that the above binomial distribution converges to a Poisson distribution. If we allow the natural interpretation of a Poisson distribution with rate parameter $0$ or $\infty$ as a point mass at 0 or $\infty$, then we have the following,

\begin{theorem}\label{thm:Poisson}
Suppose 
\begin{equation*}
\lim_{M\to\infty} M p_\ell q_i q_j = \lambda_\ell\in [0,\infty]
\end{equation*}
exists for link $\ell=(i,j)$. Then, as $M\to\infty$, the distribution of the multiplicity $W_\ell$ of $\ell$ converges to a Poisson distribution with parameter $\lambda_\ell$, i.e.
\begin{equation}\label{eq:Poisson}
\lim_{M\to\infty} \Prob[W_\ell = w] = \pi(w; \lambda_\ell)  \triangleq \frac{\lambda_\ell^w}{w!}e^{-\lambda_\ell} .
\end{equation}
\end{theorem}

The joint distribution of the $W_\ell$'s may have a complicated correlation structure. However, when  $p_\ell$ and $q_i$ scale with $M$ appropriately,  the $W_\ell$'s become asymptotically independent. 
We consider the case that $p_\ell$ and $q_i$ scale with $M$ as follows,
\begin{align}
p_\ell &\sim c_\ell M^{-\alpha_\ell}, \quad \text{for } \ell\in E, \label{eq:p}\\
q_i &\sim d_i M^{-\beta_i}, \quad \text{for } i\in V,\label{eq:q}
\end{align}
where $c_\ell, d_i > 0$ and $\alpha_\ell,\beta_i \geq 0$. For a link $\ell = (i,j)$, the parameter $\lambda_\ell$ defined in \prettyref{eq:Poisson} is then given by
\begin{equation}\label{eq:lambda}
\lambda_\ell = \begin{cases}
0, & \text{if}\quad\alpha_\ell + \beta_i + \beta_j > 1;\\
c_\ell d_i d_j, & \text{if}\quad \alpha_\ell + \beta_i + \beta_j = 1;\\
+\infty, & \text{if}\quad \alpha_\ell + \beta_i + \beta_j < 1.\\
\end{cases}
\end{equation}

For node $k$, let 
\[
\Nbr_k = \{i\in V: (k,i) \in E, \alpha_{(k,i)} + \beta_k + \beta_i = 1\}.
\]
We also assume the following condition,
\begin{equation}\label{eq:regular}
\beta_k<1, \text{ if } \; |\Nbr_k| \geq 2.
\end{equation}
 
\begin{theorem}\label{thm:asymptotic-independence}
Under the conditions \prettyref{eq:p}, \prettyref{eq:q} and \prettyref{eq:regular}, the collection of random variables $\{W_\ell:\ell\in E\}$ become asymptotically independent as $M\to\infty$, i.e.
for any $\{w_\ell:\ell\in E\}\in \N^E$, 
\begin{equation}\label{eq:asymptotic-independence}
\lim_{M\to\infty} \Prob[W_\ell = w_\ell, \forall \ell\in E] = \prod_{\ell\in E} \pi(w_\ell;\lambda_\ell)=\prod_{\ell\in E} e^{-\lambda_\ell} \frac{\lambda_\ell^{w_\ell}}{w_\ell!},
\end{equation}
where $\lambda_\ell$ is given by \prettyref{eq:lambda}.
\end{theorem}

\begin{proof}[Proof Idea]
In the large $M$ limit, each layer essentially has at most one link. Thus the configuration roughly follows a multinomial distribution, which in the limit becomes a product of Poisson distributions. The details are given in \prettyref{subsec:asymptotic-proof}.
\end{proof}

Note that condition \prettyref{eq:regular} is necessary. The example below shows that asymptotic independence does not hold when condition \prettyref{eq:regular} fails. 

\begin{example}
Consider the line network in \prettyref{fig:line-network} with three nodes. The probability that no links exist is given by
\begin{align*}
\Prob[W_{(1,2)} = W_{(2,3)} = 0] = \left[1- q_2 + q_2 (1-q_1 p_{(1,2)})(1-q_3 p_{(2,3)})\right]^M.
\end{align*}
If  $\beta_2 = 1$ and $\alpha_{(1,2)} = \beta_1  = \alpha_{(2,3)} = \beta_3 = 0$, then
\begin{align*}
\lim_{M\to\infty} \Prob[W_{(1,2)} = W_{(2,3)} = 0] &= e^{-d_2 + d_2(1-q_1 p_{(1,2)})(1-q_3 p_{(2,3)})}\\
&=\pi(0, \lambda_{(1,2)})\pi(0,\lambda_{(2,3)}) e^{\lambda_{(1,2)}\lambda_{(2,3)}/d_2}, 
\end{align*}
which shows that $W_{(1,2)}$ and $W_{(2,3)}$ are not asymptotically independent. 
\end{example}

It follows from \prettyref{thm:asymptotic-independence} and the definition of $X_\ell$ that $\{X_\ell:\ell\in E\}$ is a set of asymptotically independent Bernoulli random variables with limiting marginal distribution $\lim_{M\to\infty}\Prob[X_\ell=1] = \sum_{w\geq K} \pi(w;\lambda_\ell)$. In particular, for $K=1$, the following corollary yields Theorem 1 of \cite{Guh16} if $G$ is a tree and the conjecture therein if $G$ is a general graph. 

\begin{corollary}
Suppose $p_\ell = 1$ for all $\ell\in E$,  $q_i = dM^{-1/2}$ for all $i\in V$, and $K=1$. Then
\begin{equation} 
\lim_{M\to\infty} \Prob[X_\ell = x_\ell, \forall \ell\in E] = e^{-d^2 (|E|-|\cE|)}(1-e^{-d^2})^{|\cE|}.
\end{equation}
In the limit $M\to\infty$, the merged network $\cG$ has a giant component if $d$ exceeds the threshold $\sqrt{-\log(1-p_c)}$, where $p_c$ is the bond-percolation threshold of $G$.
\end{corollary}

Another consequence of \prettyref{thm:asymptotic-independence} is the following trichotomy. 

\begin{corollary}
Suppose $p_\ell = p\sim c M^{-\alpha}$ for all $\ell\in E$,  $q_i = q \sim dM^{-\beta}$ for all $i\in V$. Then in the limit $M\to\infty$, the network $\cG$ is
\begin{enumerate}
\item an empty network with no link, if $\alpha+2\beta>1$;
\item the entire network $G$, if $\alpha+2\beta<1$;
\item an Erd\"os-R\'enyi-like sub-network of $G$ where a link exists with probability $\sum_{w\geq K} \pi(w;cd^2)$, if $\alpha+2\beta=1$. 
\end{enumerate}
\end{corollary}

As an easy application of the results obtained in this section, consider the line network in Figure \ref{fig:line-network}. When $q=1/\sqrt{M}$
we know that links become independent as $M\to\infty$, with $p_a :=1-e^{-1}$ the (asymptotic) probability that a link
is active. The expected size of the cluster node $1$ belongs to,  including this node, is then given by 
$\bar C_{a,n}=(1- p_a^{n+1})/(1-p_a) $ and is plotted in Fig. \ref{fig-line-limit}  for $n=\{1,\ldots,20\}$.
We observe that $C_{a,n}$  converges fast w.r.t. $n$, the number of links. We have also plotted in this figure the mapping
$n\to \bar C_n(1)$ for $q=1/\sqrt{M}$ with $M=50$, which shows that making the assumption that links are independent when $q=1/\sqrt{50}$  with $M=50$ 
yields a relative error of less then $10\%$  across all values of $n\in\{1,\ldots,20\}$. 
Similarly, the expected number of active links  when  $q=1/\sqrt{M}$ and $M\to\infty$, given by $\bar L_{n,a}=n p_a$, is plotted in Fig. \ref{fig-line-limit-2} as a function of $n$. 
We have also plotted in this figure the mapping $n\to \bar L_n$ (see Section \ref{ssec:line}) for $q=1/\sqrt{M}$ with $M=40$  and $M=50$.  For $M=40$ the relative error made by approximating $ L_n$ by $\bar L_{n,a}$ does not exceed $14.1\%$ across of all values of $n\in\{1,\ldots,20\}$; it does not exceed $0.59\%$ when $M=50$.
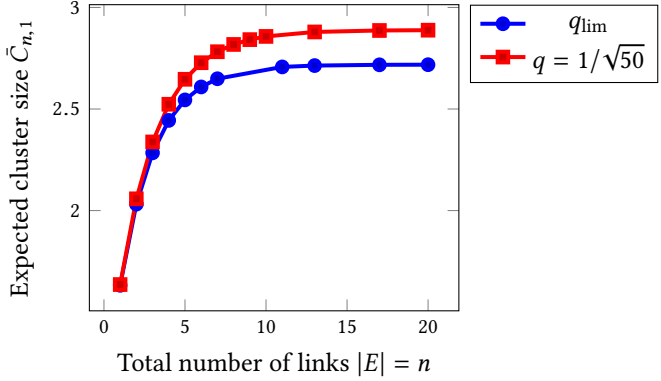
\begin{figure}[h]
\begin{center}
\begin{tikzpicture}
\begin{axis}[small, name=plot2,
   legend pos=outer north east, 
   axis on top,
   ylabel near ticks,
   xlabel near ticks,
   ylabel = {Expected cluster size $\bar C_{n,1}$},
   xlabel = {Total number of links $|E| = n$} 
]

\addplot  table [x=a, y=b] 
{
                           a      b
                         1  1.632120559
                         2  2.031696960
                         3  2.284277418
                         4  2.443938717
                         5  2.544863908
                         6  2.608660795
                         7  2.648988120
                        11  2.707218305
                        13 2.713861105
                       17 2.717576010
                        20 2.718169136
};
\addlegendentry{$q_{{\lim}}$}

\addplot  table [x=a,  y=b] 
{
                          a          b
                         1  1.635830320
                         2  2.057959327
                         3  2.337816818
                         4  2.523363723
                         5  2.646381953
                        6  2.727943470
                         7  2.782019038
                         8  2.817871329
                         9  2.841641520
                        10  2.857401240
                        13 2.879370542
                        17 2.886660730
                        20 2.888069375
};
\addlegendentry{$q=1/\sqrt{50}$}

\end{axis}
\end{tikzpicture}
\caption{Expected size of cluster containing node $1$ in the line network of \prettyref{fig:line-network} as a function of the total number of links $n$, for  $q=1/\sqrt{M}$ and $M\to\infty$ (referred to as $q_{\lim}$) and $q=1/\sqrt{50}$.}                   
\label{fig-line-limit}
\end{center}
\end{figure}
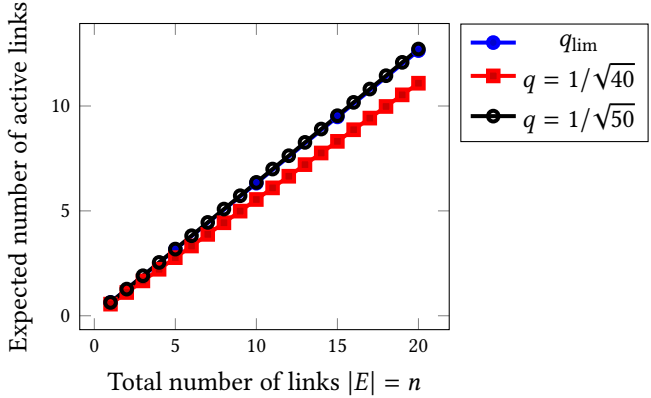
\begin{figure}[h]
\begin{center}
\begin{tikzpicture}
\begin{axis}[small, name=plot2,
   legend pos=outer north east, 
   axis on top,
   ylabel near ticks,
   xlabel near ticks,
   ylabel = {Expected number of active links},
   xlabel = {Total number of links $|E| = n$},
]

\addplot  table [x=a, y=b] 
{
                         a       b
		        1   0.6321205588
                         5    3.160602794
                      10   6.321205588
                      15  9.481808382
                        20  12.64241118
};
\addlegendentry{$q_{{\lim}}$}

\addplot  table [x=a,  y=b] 
{
                         a          b
                         1      0.5542995842
                         2      1.108599175
                         3      1.662898761
                         4      2.217198350
                         5      2.771497940
                         6      3.325797528
                         7      3.880097113
                         8      4.434396706
                         9      4.988696291
                        10      5.542995880
                        11      6.097295467
                        12      6.651595056
                        13      7.205894643
                        14      7.760194229
                        15      8.314493818
                        16      8.868793406
                        17      9.423092993
                        18      9.977392581
                        19      10.53169217
                        20      11.08599174
};
\addlegendentry{$q=1/\sqrt{40}$}

\addplot  [smooth, mark=o, black]  table [x=a,  y=b] 
{
                                a          b
                        1        0.6358288903
                         2        1.271658264
                         3        1.907487637
                         4        2.543317011
                         5        3.179146386
                         6        3.814975758
                         7        4.450805133
                         8        5.086634506
                         9        5.722463879
                        10        6.358293255
                        11        6.994122628
                        12        7.629952001
                        13        8.265781376
                        14        8.901610748
                        15        9.537440122
                        16        10.17326948
                        17        10.80909887
                        18        11.44492824
                        19        12.08075760
                        20        12.71658696
};
\addlegendentry{$q=1/\sqrt{50}$}

\end{axis}
\end{tikzpicture}
\caption{Expected number of active links in the line network of  \prettyref{fig:line-network} as a function of the total number of links $n$, for  $q=1/\sqrt{M}$ and $M\to\infty$ (referred to as $q_{{\lim}}$),  
$q=1/\sqrt{40}$,  and $q=1/\sqrt{50}$.}                   
\label{fig-line-limit-2}
\end{center}
\end{figure}
\subsection{Proof of \prettyref{thm:asymptotic-independence}}\label{subsec:asymptotic-proof}

We will use the following simple result, which holds for any probability measure.

\begin{lemma}\label{lemma:whp}
For any sequence of events $\{A_M\}$ and $\{B_M\}$ such that $\Prob[B_M] \to 1$ as $M\to\infty$, we have
\[
 \Prob[A_M] - \Prob[A_M\cap B_M] \to 0, \quad \text{as } M\to \infty.
\]
\end{lemma}

\begin{proof}
The follows from the following inequalities, 
\[
0\leq \Prob[A_M]- \Prob[A_M\cap B_M]\leq \Prob[\overline{B_M}] = 1 - \Prob[B_M].
\]
\end{proof}

Since the limiting marginal distribution \prettyref{eq:Poisson} degenerates to a point mass when $\lambda_\ell = 0$ or $\infty$, by the above lemma, we only need to prove \prettyref{thm:asymptotic-independence} for the case $\lambda_\ell\in (0,\infty)$ for all $\ell\in E$, or, 
\begin{equation}\label{eq:scaling}
\alpha_\ell + \beta_i + \beta_j = 1, \quad \forall \ell = (i,j) \in E,
\end{equation}
which we assume throughout the rest of this section. Note that this assumption and \prettyref{eq:regular} imply that $\beta_i < 1$ for all $i\in V$ with degree greater than one in $G$. Let $\Gamma = \{i\in V:\beta_i < 1\}$. If $\Gamma =\emptyset$, then every node has degree at most one, and hence \prettyref{eq:asymptotic-independence} holds trivially. We henceforth assume $\Gamma \neq\emptyset$ and let
\begin{equation}\label{eq:gamma}
\gamma = \min_{j:\in \Gamma } (1-\beta_j) \in (0,1].
\end{equation}
We can also assume that there is no isolated node in $G$ since all such nodes can be removed. 

Consider the $m$-th layer. Since the layers are i.i.d., the following analysis applies to all layers. Let $A_\ell$ be the event that link $\ell = (i,j)$ exists on this layer, i.e. $A_\ell = \{W_{m,\ell}=1\}$. The probability of this event is
\begin{equation}\label{eq:r}
r_\ell \triangleq \Prob[A_\ell] =  p_\ell q_i q_j \sim \lambda_\ell M^{-1}.
\end{equation}
Consider the probability that two links $\ell$ and $\ell'$ co-exist on the $m$-th layer. There are two cases. If $\ell$ and $\ell'$ do not share an endpoint, then $A_\ell$ and $A_{\ell'}$ are independent, so the probability is 
\[
\Prob[A_\ell \cap A_{\ell'}] = \Prob[A_\ell] \Prob[A_{\ell'}] = r_\ell r_{\ell'} \sim \lambda_{\ell}\lambda_{\ell'} M^{-2}.
\]
If $\ell$ and $\ell'$ share an endpoint $i$, then $i\in \Gamma$ and $A_\ell$ and $A_{\ell}$ are conditionally independent given $i$. Thus
\begin{eqnarray*}
\Prob[A_\ell \cap A_{\ell'}]&=&\Prob[A_\ell \cap A_{\ell'} \,|\, \hbox{ node $i$ active}]q_i \\
&=& \Prob[A_\ell \,|\, \hbox{node  $i$ active}] \Prob[ A_{\ell'}\,|\, \hbox{node $i$ active}] q _i\\
&=& \frac{r_\ell r_{\ell'}}{q_i} \sim \frac{\lambda_{\ell}\lambda_{\ell'}}{d_i} M^{-2+\beta_i} \leq \frac{\lambda_{\ell}\lambda_{\ell'}}{d_i} M^{-2+\gamma},
\end{eqnarray*}
where the last inequality follows from \eqref{eq:gamma}.
In both cases, we have
\begin{equation}\label{eq:prob-coexist}
\Prob[A_\ell \cap A_{\ell'}] \leq C M^{-1-\gamma},
\end{equation}
for some constant $C>0$ that does not depend on $\ell, \ell', m$ or $M$. 

Let $\mu_\ell$ be the probability that $\ell$ is the only  link on the $m$-th layer. We have
\begin{equation}\label{eq:mu}
\mu_\ell = \Prob\left[A_\ell \cap \bigcap_{\ell'\in E:\ell'\neq \ell} A_{\ell'}^c\right].
\end{equation}
From the identity
\[
A_\ell = \left(A_\ell \cap \bigcap_{\ell'\in E:\ell'\neq \ell} A_{\ell'}^c\right) \cup \left(\bigcup_{\ell'\in E:\ell'\neq \ell} A_\ell \cap A_{\ell'}\right),
\]
the union bound and \eqref{eq:prob-coexist}, we find
\begin{equation}\label{eq:delta}
0 \leq \delta_\ell\triangleq r_\ell - \mu_\ell \leq \sum_{\ell'\in E:\ell' \neq \ell} \Prob[A_\ell \cap A_{\ell'}] \leq |E| CM^{-1-\gamma}.
\end{equation}
Let $\mu_0$ be the probability that there is no link on the $m$-th layer, given by
\begin{equation}\label{eq:mu0}
\mu_0 = \Prob[A_\ell^c\,, \forall \ell\in E].
\end{equation}
Bonferroni's inequality 
\[
\Prob\left(\bigcup_{\ell\in E} A_{\ell}  \right)\geq \sum_{\ell\in E} \Prob(A_{\ell})-\frac{1}{2}\sum_{\ell,\ell'\in E, \ell\neq \ell'} \Prob(A_{\ell}\cap A_{\ell'})
\]
and \eqref{eq:prob-coexist} yield
\begin{equation}\label{eq:delta0}
0 \geq \delta_0 \triangleq 1-\mu_0 - \sum_{\ell\in E}  r_\ell \geq -\frac{1}{2}\sum_{\ell\neq \ell'} \Prob[A_\ell \cap A_{\ell'}] \geq -\frac{|E|^2}{2} C M^{-1-\gamma}.
\end{equation}

Consider the event that  $W_\ell=w_\ell$ for all $\ell\in E$ and each layer has at most one link. This event occurs if and only if each $\ell$ exists on a disjoint set of $w_\ell$ layers and the remaining $M-w$ layers are empty, where $w=\sum_{\ell\in E} w_\ell$. The corresponding probability is given by a multinomial distribution. Formally, if we let $B$ denote the event that each layer has at most one link, then
\begin{equation}\label{eq:multinomial}
\Prob\left[W_\ell = w_\ell, \forall \ell\in E, B\right] = \frac{M!}{(M-w)!\prod_{\ell\in E} w_\ell!} \mu_0^{M-w}\prod_{\ell\in E} \mu_\ell^{w_\ell},
\end{equation}
where $\mu_\ell$ and $\mu_0$ are defined in \eqref{eq:mu} and \eqref{eq:mu0}.

As $M\to\infty$, 
\begin{equation}\label{eq:factorial}
\frac{M!}{(M-w)!} \sim M^w.
\end{equation}
By \eqref{eq:r}, \eqref{eq:delta} and \eqref{eq:delta0}, we also have
\begin{equation}\label{eq:mu-lim}
\mu_\ell^{w_\ell} = (r_\ell - \delta_\ell)^{w_\ell} \sim  \frac{\lambda_\ell^{w_\ell}}{M^{w_\ell}},
\end{equation}
and
\begin{equation}\label{eq:mu0-lim}
\mu_0^{M-w} = \left(1-\sum_{\ell\in E} r_\ell - \delta_0\right)^{M-w} \to e^{-\sum_{\ell\in E} \lambda_\ell}.
\end{equation}
Plugging \eqref{eq:factorial}--\eqref{eq:mu0-lim} into \eqref{eq:multinomial} yields as $M\to\infty$,
\[
\Prob\left[W_\ell = w_\ell, \forall \ell\in E, B\right] \to \prod_{\ell\in E} \frac{\lambda_\ell^{w_\ell}}{w_\ell!} e^{-\lambda_\ell},
\]
which is almost the same as the claim in \eqref{eq:asymptotic-independence}.
By \prettyref{lemma:whp}, it suffices that $\Prob[B]\to 1$ as $M\to\infty$, which we show next. 

Note that the probability that a given layer has at most one link is given by $\mu_0 + \sum_{\ell\in E} \mu_\ell = 1 - \delta_0 - \sum_{\ell\in E} \delta_\ell$. Since the layers are independent, 
\[
\Prob[B]= \left(1 - \delta_0 - \sum_{\ell\in E} \delta_\ell\right)^M \sim e^{-M\delta_0 - \sum_{\ell\in E} M\delta_\ell} \to 1,
\]
where we have used \eqref{eq:delta} and \eqref{eq:delta0}. This completes the proof of \prettyref{thm:asymptotic-independence}.

\subsection{Extension to Non-identical Layers}\label{subsec:asymptotic-non-identical}

So far we have assumed that different layers are i.i.d., i.e. $p_\ell = \E[Y_{m,\ell}]$ and  $q_i = \E[Z_{m,i}]$ for all layers $m\in \sM$, $\ell\in E$ and $i\in V$. In this section, we will relax this assumption by letting $p_{m,\ell} = \E[Y_{m,\ell}]$ and  $q_{m,i} = \E[Z_{m,i}]$. The development will parallel that for the i.i.d. case.

Recall the definition $W_{m,\ell} = Y_{m,\ell} Z_{m,i} Z_{m,j}$ for $\ell=(i,j)$. Link $\ell$ exists on the $m$-th layer if and only if $W_{m,\ell}=1$. 
Define
\begin{equation}\label{eq:r}
r_{m,\ell} \triangleq \E W_{m,\ell} =  p_{m,\ell} q_{m,i} q_{m,j}.
\end{equation}

We assume that as $M\to\infty$,
\begin{equation}\label{eq:r-null}
\max_{1\leq m\leq M} r_{m,\ell} \to 0,
\end{equation}
and
\begin{equation}\label{eq:r-limit}
\E W_\ell = \sum_{m=1}^M r_{m,\ell} \to \lambda_\ell\in (0, \infty).
\end{equation}
Then again a classical result (e.g. \cite[Theorem 5.7]{kallenberg2006foundations}) shows that the marginal distribution of $W_\ell$ converges to a Poisson distribution with parameter $\lambda_\ell$ as in \eqref{eq:Poisson}.

Let $R_{\ell, \ell'}$ be the expected number of layers on which two links $\ell$ and $\ell'$ co-exist, i.e.
\begin{equation}\label{eq:R}
R_{\ell,\ell'}\triangleq  \sum_{m=1}^M \E [W_{m,\ell} W_{m,\ell'}].
\end{equation}
We replace condition \eqref{eq:regular} by the following. If two links $\ell, \ell'\in E$ share an endpoint, then as $M\to\infty$,  
\begin{equation}\label{eq:correlation}
R_{\ell, \ell'} \to 0.
\end{equation}
Note that in the i.i.d.~case, \eqref{eq:regular} implies \eqref{eq:correlation} as shown in the \prettyref{subsec:asymptotic-proof}.

\begin{theorem}\label{thm:asymptotic-independence-nonidentical}
Under the conditions \eqref{eq:r-null}, \eqref{eq:r-limit} and \eqref{eq:correlation}, the collection of random variables $\{W_\ell:\ell\in E\}$ become asymptotically independent as $M\to\infty$, i.e. \eqref{eq:asymptotic-independence} holds.
\end{theorem}

\begin{proof}
We first show that $R_{\ell, \ell'} \to 0$ as $M\to\infty$ for any pair of links $\ell\neq \ell'$. If $\ell$ and $\ell'$ share an endpoint, this is assumption \eqref{eq:correlation}.
If $\ell$ and $\ell'$ do not share an endpoint, then $W_{m,\ell}$ and $W_{m,\ell'}$ are independent, so 
\[
R_{\ell, \ell'} = \sum_{m=1}^M \E W_{m,\ell} \E W_{m,\ell'} = \sum_{m=1}^M r_{m,\ell} r_{m,\ell'}.
\]
By \eqref{eq:r-null} and \eqref{eq:r-limit},
\[
R_{\ell,\ell'} = \sum_{m=1}^M r_{m,\ell} r_{m,\ell'} \leq \max_{1\leq m'\leq M} r_{m',\ell'} \sum_{m=1}^M r_{m,\ell} \to 0,
\]
as $M\to\infty$. Thus $R_{\ell, \ell'} \to 0$ as $M\to\infty$ for all $\ell\neq \ell'$.

Let $\mu_{m,\ell}$ be the probability that $\ell$ is the only link on the $m$-th layer,  By the same argument for \eqref{eq:delta} and \eqref{eq:delta0}, we obtain
\[
0 \leq  r_{m,\ell} - \mu_{m,\ell} \leq \sum_{\ell'\in E:\ell' \neq \ell} \E [W_{m,\ell} W_{m,\ell'}].
\]
Summing over $m$, 
\[
0\leq \sum_{m=1}^M  r_{m,\ell} - \sum_{m=1}^M \mu_{m,\ell} \leq \sum_{\ell'\in E:\ell' \neq \ell} R_{\ell, \ell'}.
\]
The last sum goes to zero as $M\to\infty$. Thus \eqref{eq:r-null} and \eqref{eq:r-limit} imply $\mu_{m,\ell}\to 0$ and
\begin{equation}\label{eq:mu-m}
\sum_{m=1}^M \mu_{m,\ell} \to \lambda_\ell
\end{equation}
as $M\to\infty$.

Let $\mu_{m,0}$ the probability that there is no link on the $m$-th layer. 
Using first the same argument for \eqref{eq:delta0} and then the argument above for $\mu_{m,\ell}$, we obtain $\mu_{m,0} \to 0$ and
\begin{equation}\label{eq:mu0-m}
\sum_{m=1}^M (1-\mu_{m,0}) \to \sum_{\ell\in E} \sum_{m=1}^M r_{m,\ell} = \sum_{\ell\in E} \lambda_\ell,
\end{equation}
as $M\to\infty$.

Now we compute the Laplace transform of the random variables $W_\ell, \ell\in E$. For $t=(t_\ell:\ell\in E)$ and $t_\ell\geq 0$, the Laplace transform $\phi(t)$ is given by
\[
\phi(t) = \E \exp\left(-\sum_{\ell\in E} t_\ell W_\ell \right).
\]
Using \eqref{eq:W} and the independence between layers, we obtain
\[
\phi(t) = \E \exp\left(-\sum_{\ell\in E} t_\ell \sum_{m\in \sM} W_{m,\ell} \right)=\prod_{m\in \sM} \E \exp\left(-\sum_{\ell\in E} t_\ell  W_{m,\ell} \right).
\]
For a given $m$, 
\[
\E \exp\left(-\sum_{\ell\in E} t_\ell  W_{m,\ell} \right) \geq \mu_{m,0} + \sum_{\ell\in E} e^{-t_\ell} \mu_{m,\ell},
\]
and
\[
\E \exp\left(-\sum_{\ell\in E} t_\ell  W_{m,\ell} \right) \leq 1-\sum_{\ell\in E} \mu_{m,\ell} + \sum_{\ell\in E} e^{-t_\ell} \mu_{m,\ell}. 
\]
Using \eqref{eq:mu-m} and \eqref{eq:mu0-m}, we obtain
\[
\sum_{m=1}^M \left[1- \E \exp\left(-\sum_{\ell\in E} t_\ell  W_{m,\ell} \right)\right] \to \sum_{\ell\in E} \lambda_{\ell} (1-e^{-t_\ell}), \quad \text{as } M\to\infty.
\]
By Lemma 5.8 of \cite{kallenberg2006foundations}, as $M\to\infty$,
\[
\phi(t)\to \exp\left(\sum_{\ell\in E} \lambda_{\ell} (e^{-t_\ell}-1)\right)=\prod_{\ell\in E} \exp\left(\lambda_{\ell} (e^{-t_\ell}-1)\right).
\]
The limit is the product of the Laplace transforms Poisson distributions with parameters $\lambda_\ell$, which yields \eqref{eq:asymptotic-independence}.
\end{proof}


\section{Related Work}\label{sec:related}

As discussed in the introduction of this paper, there has been an explosion of research in multilayer networks in recent years, mostly from the physics community;  two recent review articles are~\cite{Kiv14, Boc14}.  Multilayer networks have been applied to airline networks~\cite{Du16}, transportation systems serving a common user population~\cite{Gal15}, many body problems arising in condensed matter physics~\cite{Pel95}, brain and neural networks~\cite{Bul09} and scientific collaboration networks~\cite{Bat16}, in addition to those listed in the introduction. Various models of multilayer networks relevant to different application scenarios have been proposed. These have been used in studies of diffusion dynamics of multilayer networks~\cite{Gom13}, cascades~\cite{Bru12, Gao13}, spectral properties~\cite{Rib13}, robustness ~\cite{Cel13,Byu14, Rad17}, failure mechanisms~\cite{Dom14, Rei14}, correlations~\cite{Nic15}, growing random multilayer networks~\cite{Nic13}, epidemic spread~\cite{Mar11}, community structure~\cite{Rau17}, and algorithmic complexity of finding short paths through co-evolving multilayer networks~\cite{Bas15}. The connectivity properties of random multilayer networks have also been studied, such as the study of the properties of the giant connected component (GCC) in a random network with correlated multiplexity, i.e., where the node degree distributions across layers have positive (or negative) correlations~\cite{Lee12}.  

Stochastic multilayer networks are those whose constructions can be described by one or more control parameters (such as probability of the presence of a node, edge or more complex attributes). For such networks, a wide variety of percolation formulations have been proposed and studied, e.g., competition between layers~\cite{Zha13}, weak percolation~\cite{Bax14}, $k$-core percolation~\cite{Azi14},
directed percolation~\cite{Azi14a}, spanning connectivity of a multilayer site-percolated network~\cite{Guh16}, and bond percolation~\cite{Hac16}.
Our stochastic multilayer network model can be visualized as a layered extension of the classical site-bond percolation model \cite{Ham80, Yan90}, where bonds and sites are independently occupied with probabilities $q$ and $p$. 

The work closest to ours~\cite{Guh16} considers a special case of our model where only node deactivations are permitted and these are characterized as i.i.d. Bernoulli random variables.  This work focuses on deriving conditions under which the multilayer network percolates, i.e., identifying the deactivation probability threshold such that if the deactivation probability lies below this threshold,  a single giant connected component emerges. \cite{Guh16} studied the percolation behavior as $M\rightarrow \infty$ and derived the threshold under the conjecture that links become asymptotically independent as $M\rightarrow \infty$ without proving that conjecture.  We have established this conjecture to be true with Theorem \ref{thm:Poisson} in Section \ref{sec:asymp}.  Moreover, much of our paper is concerned with characterizing and computing the multilayer network configuration distribution.


\section{Conclusions}\label{sec:concl}

In this work, we introduced a new class of stochastic multilayer networks. 
Such a network is the aggregation of $M$ random sub-networks of an underlying connectivity graph $G$. This model finds applications in social networks and communication networks, and, more generally, in any scenario where a multilayer network is formed over a common set of nodes via coexisting means of connectivity.  We showed that it is \#P hard to compute exactly and NP-hard to approximate link configuration probabilities for general $G$, and  it remains NP-hard to compute these probabilities when $G$ is a clique. 
We derived efficient recursions for computing  configuration probabilities when $G$ is a line or more generally a tree.  We showed that for appropriate scalings of the node and link selection processes
 to a layer,  link multiplicities have asymptotically independent Poisson distributions as $M$ goes to infinity.

\begin{acks}
This research was supported in part by the U.S. National Science Foundation under Grant Number CNS-1413998, and by the U.S. Army Research Laboratory and the U.K. Ministry of Defence under Agreement Number W911NF-16-3-0001. The views and conclusions contained in this document are those of the authors and should not be interpreted as representing the official policies, either expressed or implied, of the U.S. Army Research Laboratory, the U.S. Government, the U.K. Ministry of Defence or the U.K. Government. The U.S. and U.K. Governments are authorized to reproduce and distribute reprints for Government purposes notwithstanding any copyright notation hereon. This document does not contain technology or technical data controlled under either the U.S. International Traffic in Arms Regulations or the U.S. Export Administration Regulations.
\end{acks}

\bibliographystyle{ACM-Reference-Format}
\bibliography{TEX/bib}

\appendix
\section{Appendix}
\label{sec:app}
Define
\begin{align}
F(f,m)&=\sum_{i=0}^{M-m} {M-m \choose i} q^i \bar q^{M-m-i} f(i)\label{def-f}\\
G(f,m)&=\sum_{i=0}^{M-m} {M-m \choose i} q^i \bar q^{M-m-i}  \sum_{j=1}^m {m\choose j} q^j \bar q^{m-j}  f(i+j)\label{def-g}.
\end{align}
\begin{lemma}
\label{lem1}
If $f(i)=a^i b^{M-i} c$ then
\begin{align*}
F(f,m)&= c b^m (qa+\bar q b)^{M-m}\\
G(f,m)&= c(qa+\bar q b)^M - c(qa+\bar q b)^{M-m} (\bar q b)^m.
\end{align*}
\end{lemma}
\begin{proof}
This follows directly from the binomial expansion. 
\end{proof}

Note that $Q_{x^{(k)}}(m)$ in \eqref{eq:Qxnm} expresses as 
\begin{equation}
\label{Qxkm}
Q_{x^{(n)}}(m)=\bar x_{n} \bar q^m F\left(Q_{x^{(n-1)}}, m\right) + x_{n} G\left(Q_{x^{(n-1)}},m\right).
\end{equation}
\begin{lemma}
\label{lem2}
\begin{align}
&\bar q^m F\left(Q_{0^{(n-1)}},m\right)= Q_{0^{(n)}}(m) \label{FQ'k}\\
&G\left(Q_{0^{(n-1)}},m\right)=(1-(n-1) q^2 +P_{n})^M - Q_{0^{(n)}}(m). \label{GQ'k}
\end{align}
\end{lemma}

\begin{proof}
(\ref{FQ'k}) follows  from (\ref{Qxkm}) with  $x_1=\cdots=x_n=0$;  (\ref{GQ'k}) is easily obtained by using \eqref{eq:Q0nm}, (\ref{def-Pk}) and Lemma \ref{lem1}.
\end{proof}



\end{document}